\documentclass[onecolumn, draftclsnofoot,11pt]{IEEEtran}

\IEEEoverridecommandlockouts
\usepackage[english]{babel}
\usepackage[letterpaper,top=2.2cm,bottom=2.2cm,left=2.2cm,right=2.2cm,marginparwidth=1.75cm]{geometry}
\usepackage[nolist,nohyperlinks]{acronym}
\usepackage[pagewise]{lineno}
\usepackage{subfiles} 
\usepackage{titling}
\usepackage{setspace}
\usepackage{bbm}
\usepackage{physics}
\usepackage{amsthm}
\usepackage{amssymb}
\usepackage{dsfont}
\usepackage{amsmath}
\usepackage{graphicx}
\usepackage[colorlinks=true, allcolors=black]{hyperref}
\usepackage{multirow}
\usepackage[utf8]{inputenc} 
\usepackage[T1]{fontenc}    
\usepackage{url}            
\usepackage{booktabs}       
\usepackage{amsfonts}       
\usepackage{mathtools}
\usepackage{nicefrac}       
\usepackage{microtype}      
\usepackage{xcolor}         
\usepackage{wrapfig}
\usepackage{tikz}
\usepackage{quantikz}
\usepackage{enumitem} 
\usepackage{algpseudocode}
\usepackage[linesnumbered,ruled,vlined,procnumbered]{algorithm2e}
\usepackage{fix-cm} 
\usepackage{romannum}
\usepackage{cleveref}
\usepackage[super,sort&compress]{natbib}
\usepackage{float}
\usepackage{titling} 

\newtheorem{definition}{Definition}
\newtheorem{lemma}{Lemma}

\newtheorem{theorem}{Theorem}
\newtheorem{remark}{Remark}

\newtheorem*{example}{Example}
\newtheorem{thm}{Theorem}[section]

\global\long\def\RR{\mathbb{R}}

\global\long\def\EE{\mathbb{E}}
\global\long\def\II{\mathbb{I}}

\newcommand{\bfa}{\mathbf{a}}

\newcommand{\bfe}{\mathbf{e}}
\newcommand{\bfg}{\mathbf{g}}

\newcommand{\bfu}{\mathbf{u}}

\newcommand{\bfA}{\mathbf{A}}
\newcommand{\bfB}{\mathbf{B}}

\newcommand{\bfL}{L}
\newcommand{\bfV}{\mathbf{V}}

\newcommand{\bfZ}{\mathbf{Z}}

\def \yhat{\hat{y}}
\def \Yhat{\hat{\text{Y}}}

\def \tensor {\otimes}

\def \bomega {\boldsymbol{\omega}}
\def \mubar {\bar{\mu}}

\def\<{\langle}
\def\>{\rangle}

\def \tti {\texttt{i}}
\def \ttb {\texttt{b}}
\def \ttt {\texttt{t}}
\def \ttT {\texttt{T}}

\def \calB {\mathcal{B}}

\def \calE {\mathcal{E}}
\def \calF {\mathcal{F}}

\def \calH {\mathcal{H}}

\def \calM {\mathcal{M}}
\def \calN {\mathcal{N}}

\def \calR {\mathcal{R}}

\def \calX {\mathcal{X}}
\def \calY {\mathcal{Y}}

\DeclarePairedDelimiter\set{\{}{\}}
\DeclarePairedDelimiter\parens{(}{)}
\DeclarePairedDelimiter\tuple{(}{)}
\DeclarePairedDelimiter\bracks{[}{]}

\DeclareMathOperator*{\Cov}{Cov}

\newcommand{\ketphi}[1]{|\phi_{#1}\>}

\def \deq {:=}
\def \eqand {\text{ and }}
\newcommand{\add}[1]{\textcolor{black}{#1}}
\def \fubini{\textsf{F}}
\def \textth {\text{th}}
\def \Qx {\text{Q}_{\text{X}}}
\def \Qxy {\text{Q}_{\text{XY}}}
\def \Pyx {\text{P}_{\text{Y}|\text{X}}}
\def\Loss{\mathcal{L}} 
\def \nlayer {\bfL}  
\def \gradL {\nabla \Loss}
\def \state{|\phi\>}
\def \nparams {c}

\def \mqfi {\calF^{\calE}}
\def \mqfitildeij {\tilde{\calF}^{\calE}{(i,j)}}

\def \ap {\mathsf{a_p}}
\def \bq {\mathsf{b_q}}
\def \a {\mathsf{a}}
\def \b {\mathsf{b}}
\def \p {\mathsf{p}}
\def \q {\mathsf{q}}

\def \sigmaap {\sigma^{\a}_\p}
\def \sigmabq {\sigma^{\b}_\q}

\def \Sigmaap {\Sigma^{\a}_\p}
\def \Sigmabq {\Sigma^{\b}_\q}
\def \Fab {\mqfi_{[\ap,\bq]}}
\def\Zab{\bfZ_{[\ap,\bq]}} 
\def\Zbarab{\Bar{\bfZ}{\tuple{\ap,\bq}}} 
\def\Zbartij{\Bar{\bfZ}^{\ttt}{(i,j)}}
\def\Ztildeij{\Tilde{\bfZ}_{[i,j]}} 
\def \git {g_i^{\ttt}} 
\def \gjt {g_j^{\ttt}} 

\def\Loss{\mathcal{L}}
\def\param{\boldsymbol{\theta}}
\def\dparam{d\boldsymbol{\theta}}
\def\paramstar{\boldsymbol{\theta}^\star}
\newcommand{\prob}[1]{\mathbb{P}\{#1\}}
\newacro{iid}[i.i.d.]{independent and identically distributed} 
\newacro{QML}[QML]{quantum machine learning}
\newacro{QNN}[QNN]{quantum neural network}
\newacro{SGD}[SGD]{stochastic gradient descent}
\title{Quantum Natural Stochastic Pairwise Coordinate Descent 
}

\author{
\IEEEauthorblockN{Mohammad Aamir Sohail \IEEEauthorrefmark{1},  Mohsen Heidari \IEEEauthorrefmark{2}, and S. Sandeep Pradhan \IEEEauthorrefmark{1}}\\
      \IEEEauthorrefmark{1}   Department of EECS, University of Michigan, Ann Arbor, USA\\
 \IEEEauthorrefmark{2}   Department of Computer Science, Indiana University, Bloomington, USA \\
 \IEEEauthorrefmark{1} mdaamir@umich.edu 
 \IEEEauthorrefmark{2} mheidar@iu.edu
 \IEEEauthorrefmark{1} pradhanv@umich.edu
}
\date{}

\begin{document}

\maketitle
\pagenumbering{arabic}
\begin{abstract}
Variational quantum algorithms, optimized using gradient-based methods, often exhibit sub-optimal convergence performance due to their dependence on Euclidean geometry. Quantum natural gradient descent (QNGD) is a more efficient method that incorporates the geometry of the state space via a quantum information metric. However, QNGD is computationally intensive and suffers from high sample complexity. 
In this work, 
we formulate a novel quantum information metric 
and construct an unbiased estimator for this metric using single-shot measurements. 
We develop a quantum optimization algorithm that leverages the geometry of the state space via this estimator while avoiding full-state tomography, as in conventional techniques. 
We provide the convergence analysis of the algorithm under mild conditions. 
Furthermore, we provide experimental results that demonstrate the better sample complexity and faster convergence of our algorithm compared to the state-of-the-art approaches. Our results illustrate the algorithm's ability to avoid saddle points and local minima.

\end{abstract}
\section{Introduction}
The potential of quantum computing (QC) to solve complex problems exponentially faster than classical computers promises revolutionary advancements in data analysis, optimization, and encryption \cite{Elben2022,shor1994algorithms,grover1996fast,harrow2009quantum,Bacon2005}. With these rapid advancements in QC, there has been significant research interest in the development of quantum algorithms for applications involving quantum data (such as data generated by quantum many-body systems \cite{tasaki2020physics} and quantum sensors
\cite{degen2017quantum}). Examples of such applications include quantum anomaly detection \cite{liu2018quantum}, quantum state discrimination \cite{Chen2018,slussarenko2017quantum}, quantum pattern matching \cite{sasaki2002quantum}, quantum sensor networks \cite{xia2021quantum}, and
quantum data classification \cite{sentis2019unsupervised,dixit2021searching}. 

Toward this endeavor, substantial efforts have been made in designing quantum learning and optimization algorithms through a broad class of variational quantum algorithms (VQAs) \cite{farhi2018classification,Mitarai2018}, which include popular frameworks like the variational quantum eigensolver (VQE) \cite{tilly2022variational}, the quantum approximate optimization algorithm (QAOA) \cite{farhi2014quantum}, \add{\acp{QNN} \cite{Beer2020,schuld2014quest}, quantum convolutional neural networks (QCNNs) \cite{cong2019quantum}, quantum capsule networks (QCapsNets) \cite{liu2022quantum}, and quantum recurrent neural networks (QRNNs) \cite{bravo2022quantum}.} 
These models involve parameterized quantum circuits (PQC) with tunable parameters that can be \textit{learned} using quantum-classical optimization loops. Typically, gradient descent (GD) is used as a classical optimizer in VQAs due to its ease of implementation.  
However, it often suffers in terms of performance. For example, GD has a greater possibility of getting trapped at saddle points and may not converge to a local minimum \cite{dauphin2014identifying}. This is because, GD strongly relies on Euclidean or $\ell_2$ geometry within the parameter space \cite{amari1998whyng}, which poses a significant drawback because a notable distance between parameters (measured using Euclidean distance) may not necessarily have an equivalent influence on the underlying parameterized quantum states. This means that for two different sets of parameters
with a large Euclidean distance, the corresponding parameterized quantum states
may be indistinguishable. Therefore, it is preferable to 
measure the distance between parameterized quantum states and perform the steepest descent by considering the geometry of the space of these parametrized quantum states. 

In light of this, quantum natural gradient descent (QNGD) \cite{stokes2020quantum} has been proposed as a more suitable optimization method alternative to GD. QNGD moves in the steepest direction with respect to the quantum information (Riemannian) metric tensor associated with the space of quantum states, and it remains invariant under any smooth and invertible re-parameterizations \cite{amari1998natural,martens2020new}. Previous studies have shown that QNGD provides an advantage in optimizing parameterized quantum systems by taking optimization paths more aligned with the underlying geometric structure of quantum states, compared to other strategies \cite{wierichs2020avoiding}.  

Despite these advantages, using QNGD presents notable challenges in learning from quantum data for the following reasons: 
$(i)$ In the case of a density matrix (or mixed state), the quantum information metric tensor reduces to the Bures metric tensor \cite{liu2020quantum,petz1996geometries,petz1998information}, which relies on the spectral decomposition of the density matrix. Thus, evaluating the Bures metric tensor necessitates a full-state tomography of the density matrix, which in turn requires an exponential number of identical copies of the quantum state with the unknown density matrix {\cite{haah2016sample,flammia2012quantum,rath2021quantum,vitale2023estimation,gross2010quantum}}. $(ii)$ No-cloning theorem implies that an unknown quantum state cannot be replicated \cite{wootters1982single}. This means that a quantum state can only be measured once to extract any information. $(iii)$ Every iteration of QNGD requires evaluating the quantum information metric tensor, which has a cost that scales quadratically with the number of PQC parameters \add{\cite{stokes2020quantum}}. We further elaborate on these challenges and related works in Supplementary Note 1.

To the best of our knowledge, QNGD has not been thoroughly explored due to these constraints in the context of \ac{QML} for quantum data. The majority of the prior works have focused on the following aspects: {single-shot} gradient estimation \cite{heidari2022toward}, reducing the computational complexity of the quantum information metric tensor, and improving methods for its estimation \cite{gacon2021simultaneous,koczorQNGmixed,van2021measurement}. Some of these works have primarily focused on finding the ground state energy of a quantum system using VQE. Moreover, the convergence performance of QNGD in training a quantum learning model has not received enough attention despite several experiments demonstrating faster convergence than GD. 

\add{In this work, we introduce a novel metric-based quantum optimization algorithm in the framework of VQA. Our solution is based on a new ensemble-based quantum Fisher information metric (E-QFIM), in which each element can be efficiently estimated without requiring full-state tomography of the density matrix, thereby addressing the limitations of the Bures metric tensor. We construct an unbiased estimator for the full E-QFIM by first estimating a randomly selected 
$2\times 2$ submatrix using single-shot measurements. Then, we perform classical post-processing to design an unbiased estimator for the entire E-QFIM. Using the resulting unbiased estimator of E-QFIM, we design a metric-based single-shot optimization algorithm, namely, quantum natural stochastic pairwise coordinate descent (2-QNSCD). The empirical results show that 2-QNSCD outperforms single-shot stochastic gradient descent (SGD) despite both algorithms using the same number of samples. 
This approach aims to make the metric-based optimization method practical and efficient for QML applications involving quantum data, such as quantum simulation, quantum sensing, and quantum many-body physics.
}


In what follows, we present an (informal) overview of our main results, followed by a more formal and detailed discussion (Theorems \ref{thm:EQFI},\ref{thm:unbiasedQFIM}, \ref{thm:2QNCDconv}, and \ref{thm:2QNSCDsample-comp}) and leaving proofs for supplementary notes.  
Our notations are as follows: We use $[\nparams] \deq \set{1, 2, \ldots, \nparams}$. 
\add{Vectors are represented using boldface lowercase letters (including Greek letters) such as $\param$ and $\bfg$.}
Matrices are denoted by boldface uppercase letters, such as $\bfZ$. For a given matrix $\bfZ$, we denote the element of $\bfZ$ in the $i^{\text{th}}$ row and $j^{\text{th}}$ column as $\bfZ_{(i,j)}$, and $\bfZ_{[i,j]}$ denotes the $2\times 2$ submatrix of $\bfZ$ formed by selecting the $i^\textth\eqand j^\textth$ rows and columns.
For a given pure quantum state $\state$, the corresponding density matrix is denoted as $\Phi \deq |\phi\>\<\phi|$. \add{ To represent a series of quantum operations, we introduce the following shorthand notation: $1. \ W_{[\a:\b]} \deq W_{\a}W_{\a+1} \cdots W_{\b}$ and $W_{[\a:\b)} \deq W_{\a}W_{\a+1} \cdots W_{\b-1}$ when $\a \leq \b$.  $2.\ W_{[\a:\b]} \deq W_{\a}W_{\a-1} \cdots W_{\b}$ and $W_{[\a:\b)} \deq W_{\a}W_{\a-1} \cdots W_{\b-1}$, when $\a > \b$. Finally, we define the
 the covariance between two observables $\bfA$ and $\bfB$ acting on a mixed state $\rho$ as:}
 \begin{equation*}
     \add{\Cov(\bfA,\bfB)_{\rho} := \tfrac{1}{2}\text{Tr}(\{\bfA,\bfB\}\rho) - \text{Tr}(\bfA\rho)\text{Tr}(\bfB\rho).}
 \end{equation*}
\section{Results}
\noindent \textbf{Overview of the results.}
We summarize the key contributions by informally stating the main results below. 
Let $\Loss(\param)$ be a loss function associated with a PQC described by a parameterized unitary $U(\param)$, followed by a measurement (POVM) $\set{\Lambda_y}_{y\in \calY}$. The unitary $U(\param)$ acts on an input ensemble of quantum state $\set{(\Qx(x), \ket{\phi_x})}_{x\in \calX}$ with $\rho$ being the corresponding density matrix. Let $\nparams$ denote the number of PQC parameters.

 $\bullet \ \textbf{Ensemble-based Quantum Fisher Information Metric.}$ We present a new quantum information metric tensor for an ensemble of pure states, namely, \emph{ensemble-based quantum Fisher information metric} (E-QFIM), denoted as $\mqfi$. 
We derive the following characterization stated below (see Theorem \ref{thm:EQFI} for the formal statement) by constructing a new measure of closeness between ensembles, namely, ensemble fidelity, which serves as a lower bound to the Uhlmann fidelity. 
   \add{ \begin{thm}[informal]\label{thm:EQFI_informal}
    For a given PQC and an input density matrix $\rho$,
    the matrix elements of the E-QFIM are given as:
\begin{equation}\label{eqn:thm:EQFI_informal}
        \mqfi_{(k,l)}(\param) \deq 
        \Cov(\Upsilon_{k}(\param),\Upsilon_{l}(\param))_{\rho},
    \end{equation}
    \noindent where 
   $\Upsilon_{k}(\param)\deq -i(\partial_{k} U^{\dagger}(\param) )U(\param) = i U^\dagger(\param)(\partial_{k}U(\param))$ is an observable. 
\end{thm}}
\add{This new ensemble-based metric can be efficiently estimated with a sample complexity that does not depend on the system size. More formally, each element of the E-QFIM involving bounded observables $\Upsilon_{k}\eqand \Upsilon_{l}$ can be estimated within an accuracy of $\epsilon$ with probability at least 
$(1-\delta)$ by using $\Theta\left(\tfrac{1}{\epsilon^2}\log\tfrac{1}{\delta}\right)$ samples. This proof follows from the standard argument using Hoeffding’s inequality. In contrast, some existing studies, utilizing quantum state tomography \cite{flammia2024quantum,haah2016sample} and randomized measurements \cite{rath2021quantum,vitale2024robust}, suggest that the sample complexity for estimating the Bures metric grows exponentially with the system size.} 


\vspace{5pt}    
$\bullet \ \textbf{Estimation of E-QFIM.}$
\add{We construct an unbiased estimator of E-QFIM, denoted as $\bar{\bfZ}(i,j)$, by first taking into account the underlying geometry of the space of quantum states corresponding to only the pair of random parameter coordinates $(i,j)$. For estimating the $2\times 2$ submatrix corresponding to parameters coordinates $i$ and $j$, we conceptualize a sequential measurement strategy (see Lemma \ref{lem:seq_meas_anticommutator}) that employs single-shot and mid-circuit measurements to estimate the following terms using only four quantum samples: the anticommutator of two observables and the product of their expectation values. Furthermore, it does not require any additional ancilla qubits, and its gate complexity is $\Theta(1)$ to that of the PQC.} 

\add{Once this $2\times2$ estimate is obtained, it is embedded into the full 
$\nparams \times \nparams$ matrix in a highly sparse structure (all off-diagonal entries set to zero except at positions $(i,j)$ and $(j,i)$). Finally, we construct an unbiased estimator of a complete matrix by adding an appropriate regularization term and performing classical post-processing to ensure that the estimator $\bar{\bfZ}(i,j)$ is both non-singular and unbiased. This approach helps to avoid the quadratic cost of estimating $O(\nparams^2)$ terms of E-QFIM. The algorithm is summarized below (see Theorem \ref{thm:unbiasedQFIM} for the formal statement).}




    \begin{thm}[Unbiased Estimation of $\mqfi$ (informal)]\label{thm:unbiasedQFIM_informal}
    For every pair of randomly selected coordinates $(i,j)$, there exists an algorithm (\text{Algorithm \ref{alg:qfimEst}} in Methods) that takes only four training quantum data to provide an unbiased estimation of the metric, i.e., $\EE[\bar{\bfZ}(i,j)] = \mqfi$, where the expectation is over the random pair of coordinates, input quantum state, and the outcome of quantum measurements. 
    \end{thm}
   
$\bullet$ \textbf{Quantum Natural Stochastic Pairwise Coordinate Descent.} 
At each iteration, 2-QNSCD uses six training quantum data, constructs an unbiased estimator of the E-QFIM using four training quantum data as described above, and updates only two parameters according to the following update rule:
    \begin{equation}
   \param^{(\ttt+1)} = \param^{(\ttt)} - \eta_\ttt \big|\bar{\bfZ}^{\ttt}(i,j)\big|^{-1} \bfg^\ttt(i,j),
\end{equation}
    where $(i,j)$ is a pair of coordinates chosen randomly at every iteration, $\eta_\ttt$ is the learning rate, $|\bar{\bfZ}^{\ttt}| \deq ((\bar{\bfZ}^{\ttt}) ^{\dagger}\bar{\bfZ}^{\ttt})^{1/2}$ and  $\bfg^\ttt$ is an unbiased estimator of the gradient $\gradL$ constructed using two training quantum data. 
    
\add{ $\bullet$ \textbf{Convergence of 2-QNSCD.} We provide a convergence analysis of the 2-QNSCD optimization method under a mild assumption on the loss function. The following theorem summarizes the convergence rate (see Theorem \ref{thm:2QNCDconv} for the formal statement).
    \begin{thm}[Convergence of 2-QNSCD (informal)]\label{thm:2QNCDconv_informal}Under the mild assumption, 2-QNSCD with a fixed learning rate achieves an expected exponential convergence rate up to a residual error bounded by $\Delta$, 
\begin{equation}\label{eq:2QNCDconvergence_rate_informal}
     \EE[\Loss(\param^{(k)}) - \Loss(\paramstar)]
\leq (1-\sigma)^{k}r_0 +\Delta,
\end{equation}
where $r_0\deq (\Loss(\param^{(0)}) - \Loss(\paramstar))$, $\param^*$ is the global optimizer for $\Loss(\param)$, $\sigma$ proportional to $1/c^2$ is the contraction factor, and $\Delta$ is the asymptotic error bound, which depends on the assumptions on the loss function, regularization constant, and moments of the estimators. 
\end{thm}
\emph{Sample Complexity}: Assuming $\epsilon>\Delta$, the number of iterations needed to ensure $\mathbb{E}[\mathcal{L}(\boldsymbol{\theta}^{(k)})-\mathcal{L}(\boldsymbol{\theta}^*)] \leq \epsilon$ is given by 
\(
k \geq \!\left(\frac{1-\sigma}{\sigma}\ln\!\left(\frac{r_0}{\epsilon - \Delta}\right)\right).
\)
Since we are using a fixed number of samples at each iteration, the overall sample complexity is 
\[O \!\left(c^2\ln\!\left(\frac{r_0}{\epsilon - \Delta}\right)\right).\]
Note that the overall sample complexity depends on $c^2$ and $\Delta$ (see Theorem \ref{thm:2QNSCDsample-comp} for the formal statement). 
The factor $c^2$ corresponds to the penalty for the 2-coordinate approximation of the full metric, and the parameter $\Delta$ represents the penalty for using estimates of the metric and the gradient in place of the original quantities. The latter is not unique to 2-QNSCD and varies depending on the particular estimation technique used. 
\begin{remark}
    For a fixed total sample budget, there is a trade-off among the following four factors: $(i)$ the contraction factor $\sigma$, $(ii)$ the asymptotic error bound $\Delta$,  $(iii)$ the number of iterations (updates) performed,
$(iv)$ condition for convergence (Assumption 2 of Theorem 3). 
In our case, we are using six samples to estimate the gradient and the metric and correspondingly update the parameters. Assume that samples are generated at a constant rate of time. 
Now, if we can spend $n$ samples per update of the parameters, then we can expect the following: (i) a smaller asymptotic error bound $\Delta'$, because of smaller estimation errors. 
(ii) a larger contraction factor $\sigma'$ because more gradient and metric elements are being estimated. (iii) a fewer (by factor $n/6$) number of iterations (parameter updates) as compared to 2-QNSCD, since we must collect \( n \) samples for a parameter update, while 2-QNSCD updates the parameters using just six samples. (iv) a milder condition for convergence (a relaxation on Assumption 2 of Theorem 3). As a result, we get the following inequality:
\[\mathbb{E}[\mathcal{L}(\boldsymbol{\theta}^{(6k/n)})-\mathcal{L}(\boldsymbol{\theta}^*)] \leq (1-\sigma')^{(6k/n)} \ r_0 + \Delta'.\]
\end{remark}}


$\bullet$ \textbf{Empirical Validation.} We support our theoretical findings with experimental results on a binary classification problem. Our evaluations, using multiple PQCs, demonstrate 2-QNSCD's data efficiency and faster convergence rate compared to the single-shot SGD method, which potentially converges to a local minimum or gets trapped at saddle points. 
 
     $\bullet$ \textbf{QGI Inequality.} In the pursuit of the 2-QNSCD convergence analysis, we also characterize a sufficient condition required for the exponential convergence rate of standard QNGD (see Theorem A.1 in Supplementary Note 2), assuming complete access to the metric tensor and the gradient. Particularly, we introduce the \emph{quadratic geometric information} (QGI) inequality (Def. \!4 in Supplementary Note 2), which serves as a more general metric-dependent criterion compared to the Polyak-Łojasiewicz (PL) inequality \cite{polyak1963gradient} (Def. \!3 in Supplementary Note 2) and ensures an exponentially faster rate of convergence. In Supplementary Notes 2 and 3, we provide examples of non-convex functions that satisfy the QGI inequality but do not meet the PL inequality. This illustrates that QNGD can achieve exponential convergence even when GD might not guarantee convergence.

     Note here that if a loss function satisfies the PL inequality, GD achieves an exponential convergence rate to the global minimum. It was first shown by Polyak in 1963; for more details, please refer to \cite{karimi2016linear,polyak1963gradient}. 

In the following subsections, we formally present the main results of this paper. We begin by discussing the PQC architecture considered in this work and then formalize the QML model involving quantum data under the supervised learning framework. 

\noindent \textbf{PQC architecture:} We consider a standard $d$-qubit multi-layered parameterized unitary operator 
as our PQC to train a quantum learning model. Each layer is composed of a parameterized unitary operator, denoted as $U_{\a}(\param_\a)$, followed by a non-parameterized (fixed) unitary operator, denoted as $V_\a$, where $\a$ is the layer index, and $\param_\a \in [0,2\pi)^{d}$ is a parameter vector corresponding to the layer $\a$. For simplicity of exposition, the parameterized unitaries consist of tensor products of single-qubit Pauli rotations,
i.e., 
$U_\a(\param_\a) \deq 
\bigotimes_{\p=1}^{d} R_{\sigmaap}(\param_{(\a,\p)}),$ where $\p$ denotes the qubit index, $\param_{(\a,\p)} \in [0,2\pi)$ is the parameter corresponding to the ($\a^\textth$ layer, $\p^\textth$ qubit) unitary operator, $\sigmaap$ are Pauli operators, and $R_{\sigmaap}(\param_{(\a,\p)})\deq \exp(-i\param_{(\a,\p)} \sigmaap/2)$ is the rotation operator corresponding to $\sigmaap$. Thus, mathematically, our PQC with $\nlayer$ layers is a parameterized unitary operator, defined as: 
\begin{equation}\label{eq:utheta}
        U(\param) \deq V_{\nlayer}U_{\nlayer}(\param_\nlayer)\cdots V_2U_2(\param_2)V_1U_1(\param_1),
\end{equation}
where  $\param \in [0,2\pi)^{\nparams}$ represents a vector comprising all the parameters across each layer and $\nparams = dL$ is the number of PQC parameters. 
For conciseness, we use the shorthand \(W_{\alpha}(\param_{\alpha}) \deq V_{\alpha}U_{\alpha}(\param_{\alpha})\) throughout the paper and omit the explicit parameter dependence (e.g., writing \(W_\alpha\) instead of \(W_\alpha(\param_\alpha)\)) when it is clear from context.

\begin{figure}
    \centering
    \includegraphics[height=3.25in,width=4.25in]{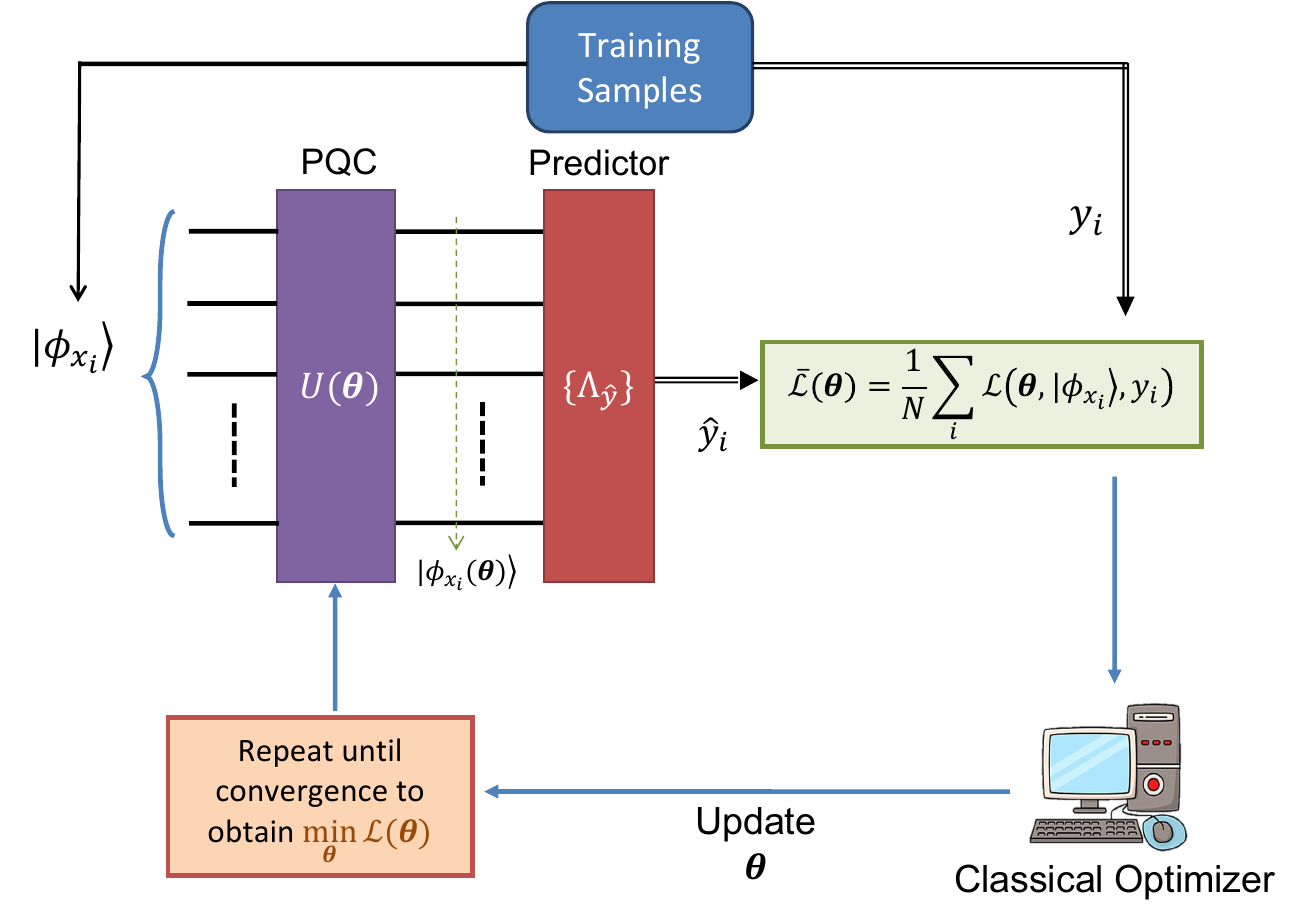}
    \caption{{Quantum Learning Model.} Quantum data labeled by classical variables is processed using a parameterized quantum circuit 
$U(\param)$, followed by a POVM $\set{\Lambda_{\yhat}}$ to predict labels. The model aims to find an optimal parameter
by minimizing the average per-sample expected loss $\bar{\Loss}(\param)$.}
    \label{fig:QLmodel}
\end{figure}
\noindent \textbf{Quantum Learning Model:} We focus on a supervised learning model with quantum data labeled by a classical variable (see Fig.~\ref{fig:QLmodel}). Consider an unknown but fixed probability distribution $\Qxy$ on $\calX \times \calY$, where $\calX \eqand \calY$ are finite sets. 
The feature set is an ensemble of pure quantum states $\calE \deq \set{(\Qx(x),\ketphi{x})}$, where $\rho \deq \sum_{x \in \calX} \Qx(x) |\phi_x\>\<\phi_x|$ represents the corresponding density matrix, and $\calY$ represents the set of all possible classical labels. 

We are given $n$ \ac{iid} samples $\set{(\ketphi{x_i},y_i)}_{i \in [n]}$, where $(x_i,y_i)$ are randomly generated according to $\Qxy$. A supervised quantum classification learning algorithm processes the training samples using $U(\param)$, followed by a quantum measurement characterized via a POVM (positive operator-valued measure)  $\Lambda \deq \set{\Lambda_{\yhat} \geq 0 : \sum_{\yhat \in \calY } \Lambda_{\yhat} = I}$ (referred to as predictor) to predict the labels of unseen samples. The goal of the learning algorithm is to find an optimal $\param^*$ that ensures the predicted label closely matches the given training labels for each training quantum sample.

\noindent \textit{Loss function}: Given a sample $\ketphi{x}$, the model generates a label $\yhat$ as the output of the measurement according to a conditional probability distribution, as determined by Born's rule: 
\begin{equation}
    \Pyx(\yhat | \ketphi{x},\param) = \Tr\{\Lambda_{\yhat}U(\param)\Phi_xU^{\dagger}(\param)\}.
\end{equation}
Let $\Yhat$ denote the corresponding random variable.
To measure the incurred prediction loss, we employ a loss function $\ell:\calY\times\calY \rightarrow[0,\infty)$, for example, 0-1 loss, absolute loss, and squared loss.
Then, conditioned on a fixed sample $\parens{\ket{\phi_x},y}$, the per-sample expected loss taken with respect to $\Yhat$ is expressed as:
\begin{equation}\label{eqn:persampleloss}
    \Loss(\param,\ketphi{x},y) \deq \EE_{\Pyx}[\ell(y,\Yhat)] = \sum_{\yhat\in \calY} \ell(y,\yhat) \Tr\{\Lambda_{\yhat}U(\param)\Phi_xU^{\dagger}(\param)\}.
\end{equation}
After taking the expectation over the sample's distribution, we get the incurred expected loss as a function of $\param$: $\Loss(\param) \deq \EE_{\text{Q}_{\text{X}\text{Y}}}[\Loss(\param,\ketphi{\text{X}},\text{Y})].$ 
Since the underlying distribution $\Qxy$ is unknown, directly minimizing the expected loss is not feasible. Consequently, given the sample set $\set{\parens{\ketphi{x_i},y_i}}_{i\in [n]}$, we aim to minimize the incurred average per-sample expected loss:
\begin{equation}\label{eqn:avgpersamplexpectedloss}
    \bar{\Loss}(\param) \deq \frac{1}{n}\sum_{i=1}^n 
    \Loss(\param,\ketphi{x_i},y_i) =  \frac{1}{n}\sum_{i=1}^n \sum_{\yhat \in \calY} 
    \ell(y_i,\yhat) \Tr\set{\Lambda_{\yhat}U(\param)\Phi_{x_i}U^{\dagger}(\param)}.
\end{equation}

We now present the main findings of our work. Firstly, we establish the framework for constructing E-QFIM and then derive the E-QFIM for our PQC model. Subsequently, we will introduce our 2-QNSCD optimization method, its convergence analysis, and experimental results demonstrating the 2-QNSCD's data efficiency and faster convergence rate.

\noindent \textbf{Ensemble Quantum Fisher Information Metric (E-QFIM).} For a given ensemble of quantum states $\calE \deq \set{\Qx(x),\ketphi{x}}_{x\in \calX}$, define a parameterized ensemble $\calE(\param) \deq \set{\Qx(x), \ket{\phi_x(\param)}}_{x\in\calX}$ with corresponding density matrix $\rho(\param) \deq \sum_x\Qx(x) |\phi_x(\param)\>\<\phi_x(\param)|$, where $\ket{\phi_x(\param)} \deq U(\param)\ketphi{x}.$ The central idea behind the E-QFIM is constructing a distance measure that avoids any use of density matrices and can be calculated solely using the ensemble representation. To achieve this, we establish a lower bound on the Uhlmann fidelity such that the resulting expression is free from the density matrix representation.
Consider the following inequalities:
\begin{align}\label{eqn:lowerbound_UhlmannFidelity}
    f_\rho(\param,\param^{'}) 
& \geq \sum_x\Qx(x)|\<\phi_x(\param)|\phi_x(\param^{'})\>|^2 \geq \big|\sum_x\Qx(x)\<\phi_x(\param)|\phi_x(\param^{'})\>\big|^2, 
\end{align}
where $f_\rho(\param,\param^{'}) \deq (\Tr\{({\sqrt{\rho(\param)}\rho(\param^{'})\sqrt{\rho(\param)}})^{1/2}\})^2$ is the Uhlmann fidelity \cite{wilde2013quantum}, the first and second inequalities follow from the joint concavity of Uhlmann fidelity \cite[Ch.9]{wilde2013quantum} and Jensen's inequality, respectively. It can be observed that the lower bound expression above can be described using only the ensemble of pure states. 
 Now, we select the lower bound expression of \eqref{eqn:lowerbound_UhlmannFidelity} as a measure of closeness between ensembles of quantum states. We define \emph{ensemble fidelity} between two ensembles of pure states characterized by parameters $\param \eqand \param^{'}$ as 
\begin{equation}
    \label{eqn:ensemble_fid}
    f_\calE(\param,\param^{'})\deq \Big|\sum_{x\in \calX} \Qx(x) \<\phi_x(\param)|\phi_x(\param^{'})\>\Big|^2.
\end{equation} 
Subsequently, define a distance measure between ensembles of quantum states using \eqref{eqn:ensemble_fid} as follows:
\begin{equation}\label{eqn:ensemble_dist}
    d_{\text{E}}(\calE(\param),\calE(\param^{'})) \deq \big(2-2\sqrt{f_\calE(\param,\param^{'})}\big)^{1/2}.
\end{equation}

 \noindent Using \eqref{eqn:lowerbound_UhlmannFidelity}, it can be easily observed that \eqref{eqn:ensemble_dist} serves as an upper bound to the Bures distance (see Supplementary Note 1). The ensemble distance measure is useful as a measure of the distinguishability of two 
ensembles $\calE(\param)\eqand\calE(\param^{'})$ because 
it is non-negative and 
 equals zero if and only if every corresponding quantum state within $\calE(\param)$ and $\calE(\param^{'})$ is essentially the same, i.e., differing solely by a constant global phase: $|\phi_x(\param^{'})\>= e^{i\delta}  \ket{\phi_x(\param)}$ for all $x \in \mathcal{X}$, where $\delta \in [0,2\pi)$. Furthermore, the distance measures $d_{\text{E}}$ is monotonic, i.e., it does not increase under the action of any quantum channel (\add{Completely Positive Trace-Preserving} map \cite{wilde2013quantum}). We discuss these in more detail in Supplementary Note 4.
 
In order to understand how a change in the parameter $\param$ affects the underlying ensemble, we examine how a slight perturbation of the parameter $(\param + \dparam)$ reflects in the ensemble distance. We develop this by analyzing the ensemble distance between two infinitesimally close ensembles  
using the Taylor series expansion around $\dparam = 0$. 
Noting the simple fact that the ensemble distance measure achieves its minimum for $\dparam=0$, i.e., $d_{\text{E}}(\calE(\param),\calE(\param)) = 0$, we infer that the first-order contribution vanishes around the minimum, and the second-order term will be the first contribution of the Taylor series expansion.
The squared infinitesimal ensemble distance is expressed as:
\begin{equation*}
    \text{d}s^2 := d_{\text{E}}^2(\calE(\param),\calE(\param+\dparam)) = \dparam^{\ttT}\mqfi(\param)\dparam + O(\|\dparam\|^3),
\end{equation*}
where we refer to $\mqfi(\param)$ as E-QFIM, whose entries are given as, 
\begin{equation*}
    \mqfi_{(\ap,\bq)}(\param) = \frac{1}{2}\frac{\partial^2 }{\partial \dparam_{(\a,\p)} \partial \dparam_{(\b,\q)}} d_{\text{E}}^2(\param,\param + \dparam) \Big\vert_{\dparam = 0}.
\end{equation*}
Here, we use the short-hand notation $\ap \deq (\a,\p)$. The information metric $\mqfi(\param)$ essentially encapsulates all the information about the local vicinity of $\param$, measured by $d_{\text{E}}$ in the space of quantum states. The E-QFIM can be calculated by the following theorem:

\begin{theorem}\label{thm:EQFI}
    For a given unitary operator $U(\param)$ and ensemble $\calE = \set{(\Qx(x),|\phi_{x}\>)}$,
    the elements of the ensemble quantum Fisher information metric (E-QFIM) are given as:
    \begin{equation}\label{eqn:thm:EQFI}
        \mqfi_{(\ap,\bq)}(\param) \deq \Re{\Tr(\Upsilon_{\ap}(\param)\Upsilon_{\bq}(\param)\rho) - \Tr(\Upsilon_{\ap}(\param)\rho)\Tr(\Upsilon_{\bq}(\param)\rho)} = \Cov(\Upsilon_{\ap}(\param),\Upsilon_{\bq}(\param))_{\rho},
    \end{equation}
    where $\rho \deq \sum_{x\in\calX} \Qx(x) |\phi_{x}\>\<\phi_{x}|$ is the density matrix corresponding to the ensemble $\calE$ and $\Upsilon_{\ap}$ is a Hermitian operator defined as \cite{taddei2013quantum}:
    \begin{equation*}
        \Upsilon_{\ap}(\param)\deq -i\left(\frac{\partial U^{\dagger}(\param)}{\partial \param_{(\a,\p)}}\right)U(\param) = i U^\dagger(\param)\left(\frac{\partial U(\param)}{\partial \param_{(\a,\p)}}\right).
    \end{equation*}
\end{theorem}

 A detailed proof is provided in Supplementary Note 5. It can be easily seen that the $\mqfi$ metric reduces to the Fubini-Study metric (see Def. \!1 in Supplementary Note 1) for pure quantum states, i.e., $\calE = \set{(1,\ket{\phi})}$. In this sense, E-QFIM can be considered a generalization of the Fubini-Study metric to an ensemble of quantum states. We derive an explicit expression of E-QFIM for our PQC architecture, as outlined in Lemma \ref{lem:QNN_EQFI}. The complete derivation is provided in Supplementary Note 6.

\begin{lemma}\label{lem:QNN_EQFI} For the PQC $U(\param)$ described in Results and an ensemble $\set{(\Qx(x),|\phi_x\>)}$ with the corresponding the density matrix $\rho$, the entry of E-QFIM is obtained as 
    \begin{align}\label{eqn:QNN_EQFI}
    \mqfi_{(\ap,\bq)}(\param) = \Cov(\Upsilon_{\ap},\Upsilon_{\bq})_{\rho}
\end{align}
where $\Upsilon_{\ap} \deq W_{[1:\a)}^{\dagger}\parens{\II^{\tensor [1:\p)}\tensor (\sigmaap/2) \tensor \II^{\tensor (\p:d]}} W_{(\a:1]}$ is a Hermitian matrix, $\Upsilon_{\ap}^2 = {\Upsilon_{\bq}^2} = \II/4$.
\end{lemma}
 For brevity, we omit the dependence of $\Upsilon_{\ap}$ and $\Upsilon_{\bq}$ on $\param$, as this is clear from the context.

\noindent \noindent\textbf{Estimation of E-QFIM.}
Looking at the expression of E-QFIM, we see it requires designing estimators for $\text{Tr}(\Upsilon_{\ap}\rho)$, $[\text{Tr}(\Upsilon_{\ap}\rho)]^2$,  and $\text{Tr}(\set{\Upsilon_{\ap},\Upsilon_{\bq}}\rho)$ for all coordinates $\ap$ and $\bq$. 
One can see that the estimation for the first term is relatively simple, as $\Upsilon_\ap$ is observable by itself. For the next term, it amounts to creating an estimator for a term of the form $(\EE[A])^2$, where $A$ is a random variable with a probability distribution $P_A$. To accomplish this, we can generate a pair of independent random variables $A_1$ and $A_2$ with the distribution $P_A$. By taking the product $A_1 A_2$, we obtain an unbiased estimator for $(\EE[A])^2$. In this context, the measurement outcomes are treated as independent random variables, and their product is used to form the desired estimator.

{For the last term, note that $\{\Upsilon_\ap,\Upsilon_\bq\}$ is an observable. However, we do not construct a circuit for the anti-commutator directly because it may not admit an explicit decomposition in terms of standard quantum gates. Even if such decomposition exists, there is no assurance that the resulting implementation would have practical or efficiently realizable depth. 
To address these challenges, we devise a sequential measurement strategy, which leverages the inherent interdependence among these terms, all of which involve $\Upsilon_\ap$ and $\Upsilon_\bq$ in different forms. The measurement strategy is outlined below in Lemma \ref{lem:seq_meas_anticommutator}, and a detailed proof can be found in Supplementary Note 7.}

\begin{lemma}[Sequential Measurement for Anti-Commutator]\label{lem:seq_meas_anticommutator}
Consider Hermitian matrices $\bfA$ and $\bfB$ on Hilbert space $\calH$, such that $\bfA^2 = \II$, i.e., $\bfA$ has only two eigenvalues $\pm1$, and the eigenvectors corresponding to these eigenvalues form a complete orthonormal basis for $\calH$. 
Given a quantum state $\rho$, the expectation of anti-commutator $\set{\bfA,\bfB}$, i.e., $\Tr(\set{\bfA,\bfB}\rho)$, can be computed using the following sequential measurement strategy:
\begin{enumerate}[label=\normalfont \Roman*.]
    \item Perform a measurement on the state $\rho$ along the eigenvectors of $\bfA$, given as $\calM \deq \{\bfA_+,\bfA_-\}$,
    where $\bfA_+$ and $\bfA_-$ are projectors onto eigenspace of $\bfA$ corresponding to $+1$ and $-1$ eigenvalues, respectively. 
    \item Measure the collapsed (post-measurement) state along the eigenvectors of $\bfB$.
\end{enumerate}
Suppose $A$ and $B$ are random variables denoting the measurement outcomes of Steps  
 $\Romannum{1}$ and $\Romannum{2}$, respectively. Then the following equality holds: $\frac{1}{2}\Tr(\set{\bfA,\bfB}\rho) = \EE[AB]$. 
\end{lemma}


Let $\bar{\bfZ}(\ap,\bq)$ be a random $\nparams\times \nparams$ matrix with the following sparse structure. All the off-diagonal entries of the matrix $\bar{\bfZ}$ are zero except the entries corresponding to coordinates $(\ap,\bq)\eqand(\bq,\ap)$. For brevity, we drop the dependence of $\nparams\times\nparams$ matrices, such as $\bar{\bfZ}$, on the pair of coordinates $(\ap, \bq)$ when it is clear from the context. Now, to construct $\bar{\bfZ}$, we begin by randomly selecting a pair of coordinates, say $\ap \eqand \bq$, and then construct a sparse $\nparams\times \nparams$ random matrix, denoted as $\bfZ(\ap,\bq)$. All entries of $\bfZ$ are zero except for the ones corresponding to its $2\times 2$ submatrix $\Zab$.
The seed of the idea for constructing the elements of $\Zab$ is introduced above, and further details are provided in the Methods section. Furthermore, as outlined in Lemma \ref{lem:2x2estimateQFIM}, $\Zab$ needs to be an unbiased estimator of the $2\times 2$ sub-matrix $\mqfi_{[\ap,\bq]}$. The proof is provided in Supplementary Note 8.

\begin{lemma}[Unbiased Estimator of $\Fab$]\label{lem:2x2estimateQFIM}
    For a given pair of coordinates $(\ap, \bq)$, let $\Zab$ be the matrix generated using Algorithm \ref{alg:qfimEst} in Methods. Then, $\EE\left[\Zab\right]=\Fab,$ where the expectation is over measurement outcomes and input quantum states.   
\end{lemma}

Next, we perform classical post-processing on $\bfZ$ to construct a $\nparams\times \nparams$ matrix, denoted as $\Tilde{\bfZ}(\ap,\bq)$. \add{The structure of the $\nparams\times\nparams $ matrix $ \Tilde{\bfZ}$ is as follows: all entries zero except for the corresponding to the $2\times2$ submatrix $\Tilde{\bfZ}_{[\ap,\bq]}$, which is a scaled and shifted version of ${\bfZ}_{[\ap,\bq]}$.} We add a (scaled) regularization constant $\beta$ to the diagonal terms of $\Tilde{\bfZ}$ to ensure the resulting unbiased estimator is numerically stable and positive-definite. Finally, the unbiased estimator of E-QFIM is given as,
\begin{equation*}
     \bar{\bfZ}(\ap,\bq) \deq \frac{\nparams(\nparams-1)}{2}\Big(\Tilde{\bfZ} - \frac{2\beta}{\nparams}\II\Big),
 \end{equation*}
 where $\II$ denotes $\nparams\times \nparams$ identity matrix. The following theorem shows that the final matrix $\bar{\bfZ}(\ap,\bq)$ is an unbiased estimator of $\mqfi$. The proof is provided in Supplementary Note 9.

\begin{theorem}[Unbiased Estimation of $\mqfi$]\label{thm:unbiasedQFIM}
    Let $\bar{\bfZ}{(\ap,\bq)}$ be the matrix generated from Algorithm \ref{alg:qfimEst} (see Methods) for a pair of randomly selected coordinates  $(\ap,\bq)$. The expectation of $\bar{\bfZ}{(\ap,\bq)}$ over the choice of a pair of coordinates, measurement outcomes, and input quantum states satisfies: \begin{equation*}
        \EE[\bar{\bfZ}{(\ap,\bq)}] 
    =\mqfi. 
    \end{equation*}
\end{theorem}

\noindent \textbf{Estimation of the Gradient.}
We begin by computing the derivative of the per-sample expected loss with respect to $\param_{(\a,\p)}$, which can be expressed as \begin{equation*}
    \partial_{\ap}\Loss(\param,\ketphi{x},y) = \sum_{\yhat \in \calY} -\frac{i}{2}\ell(y,\yhat) \Tr\{\Lambda_{\yhat}W_{[L:a]}\big[\Sigmaap, \Phi^\a_x\big] W^{\dagger}_{[a:L]}\},
\end{equation*}
where $\Phi^\a_x \deq W_{(\a:1]}{\Phi_x}W_{[1:\a)}^{\dagger}\eqand \Sigmaap \deq (\II^{\tensor [1:\p)} \tensor \sigmaap \tensor \II^{\tensor (\p:d]})$. We provide the derivation in Supplementary Note 10. Here, note that the derivative involves a commutator term, which poses a challenge because it is an anti-Hermitian matrix and cannot be measured directly. Addressing this challenge, Heidari et al.  \cite{heidari2022toward} introduced an innovative approach to constructing an unbiased estimate for the commutator and, hence, for the gradient. Below, we present a lemma summarizing the gradient estimation method proposed in Lemma 1 (Heidari et al.). For completeness, a detailed proof is given in Supplementary Note 11.

\begin{lemma}[Observable for Commutator]\label{lem:obs_commutator}
  Consider Hermitian matrices $\bfA$ and $\bfB$ on Hilbert space $\calH$, such that $\bfA^2 = \II$, i.e., $\bfA$ has only two eigenvalues $\pm1$, and the eigenvectors corresponding to these eigenvalues form a complete orthonormal basis for $\calH$. For a given quantum state $\rho \in \calH$, the expectation of commutator $\bracks{\bfA,\bfB}$ can be expressed as 
  $\Tr(\bracks{\bfA,\bfB}\rho) =2i \Tr(O\bfV \tilde{\rho} \bfV^{\dagger}),$ where $\tilde{\rho} \deq \rho \tensor |+\>\<+|$, $O$ is an observable, and $\bfV$ is a unitary operator defined as: 
  \begin{equation*}
      O \deq \bfB\tensor \ketbra{0} - \bfB\tensor\ketbra{1} \ \eqand \ \bfV \deq e^{i\pi \bfA/4} \tensor \ketbra{0} + e^{-i\pi \bfA/4} \tensor \ketbra{1}.
  \end{equation*}   
\end{lemma}

To build an unbiased estimator of $\gradL$, we first randomly select a pair of coordinates, say $\ap$ and $\bq$ and then construct a random $\nparams\times 1$ vector, denoted as $\bfg(\ap,\bq)$, with the following sparse structure:  It contains only two non-zero elements corresponding to the chosen coordinates \(\ap\) and \(\bq\). The seed of the idea for constructing these elements using one quantum sample each is outlined in Lemma \ref{lem:obs_commutator}, and additional details are provided in the Methods section. The following theorem shows that $\bfg$ is an unbiased estimator of $\gradL$, and the proof is provided in Supplementary Note 12.

\begin{lemma}[Unbiased Estimation of $\gradL$]\label{thm:unbiasedgrad}
    Let ${\bfg}{(\ap,\bq)}$ be the vector generated from Algorithm \ref{alg:gradientEst} (see Methods) for a pair of randomly selected coordinates  $(\ap,\bq)$. The expectation of ${\bfg}{(\ap,\bq)}$ over the choice of the pair of coordinates, measurements outcome, and input quantum states satisfies: \begin{equation*}
        \EE[{\bfg}{(\ap,\bq)}] =\gradL.
    \end{equation*}
\end{lemma}

\noindent \textbf{Quantum Natural Stochastic Pairwise Coordinate Descent.}
We now introduce the novel metric-based 2-QNSCD optimization algorithm with the following update rule:
\begin{equation}\label{eq:2QNCD update}
  \param^{(\ttt+1)} = \param^{(\ttt)} - \eta_\ttt |\Bar{\bfZ}^{\ttt}{\tuple{\ap,\bq}}|^{-1}\bfg^{\ttt}{\tuple{\ap,\bq}}.
\end{equation}
Here, $\mathrm{\ap}$ and $\bq$ are distinct coordinates randomly chosen from the set $[\nparams]$ and 
$\eta_\ttt$ denotes the learning rate.
In this approach, there are three levels of stochasticity: $(a)$ the selection of the input quantum state, $(b)$ the selection of the coordinate pair, and $(c)$ the measurement outcome of a quantum circuit.

At each iteration, 2-QNSCD selects a random pair of coordinates, denoted as $\ap$ and $\bq$, and constructs the random vector $\bfg^{(\ttt)}{(\ap, \bq)}$ and the random matrix $\Bar{\bfZ}^{\ttt}{\tuple{\ap,\bq}}$.
Our estimation process has a constant overhead, as it only uses six samples per iteration of the {2-QNSCD} algorithm: four for the metric and two for the gradient. Finally, we emphasize that in the update rule \eqref{eq:2QNCD update} only two coordinates of the PQC parameter $\param$, particularly, $\param_{(\a,\p)}\eqand\param_{(\b,\q)}$ are updated. We provide additional details in the Methods section and summarize the 2-QNSCD algorithm in Algorithm \ref{alg:2-QNSCD}.
\begin{algorithm}[ht]
\caption{2-QNSCD}
\label{alg:2-QNSCD}
\DontPrintSemicolon
\setstretch{1.25}
\LinesNumbered
\SetKwFunction{FSum}{Gradient\_Estimator}
\SetKwFunction{FSub}{E-QFIM\_Estimator}

\KwIn{Training data $\{(\ketphi{\ttt}, y_{\ttt})\}_{\ttt=0}^{(6n-1)},$ learning rate $\eta_{\ttt}$, and regularization constant $\beta >0$}
\KwOut{Updated PQC parameters: $\param^{(n)}$}
\tcc{Initialization}
Randomly select the parameters $\param^{(0)}$ over $[0,2\pi)^\nparams$\;
\For{$ \normalfont \ttt = 0$ to $(n-1)$}{
    Randomly select a pair of coordinates $(\ap,\bq)$\; 
    $\bfg^\ttt(\ap,\bq) = $ \FSum{$\{(\ketphi{(6\ttt+k)},y_{(6\ttt+k)}),\ap,\bq\}_{k \in \set{0,1}}$}\;
    $\bar{\bfZ}^{\ttt}{(\ap,\bq)} =$ \FSub{$\set{\ketphi{(6\ttt+k)}}_{k \in [2:5]},\ap,\bq$}\;
    Update the parameter as :\;
    $\param^{(\ttt+1)} \leftarrow \param^{(\ttt)} - \eta_{\ttt} |\bar{\bfZ}^{\ttt}(\ap,\bq)|^{-1}\bfg^\ttt(\ap,\bq)$\;
}
\KwRet $\param$
\end{algorithm}

\noindent \add{\textbf{Convergence Analysis of 2-QNSCD.}
We first state the following assumptions on the loss function $\Loss(\param)$.}

\add{\noindent\textit{Assumption 1.} (Pairwise $\mathsf{L}_2$-smooth) The function $\Loss(\param)$ is pairwise $\mathsf{L}_2$-smooth (or has a pairwise $\mathsf{L}_2$-Lipchitz continuous gradient), i.e., for all pair of coordinates $(i,j) \in [c]^2, 
 \param \in \RR^c, \eqand \alpha_i, \alpha_j \in \RR$, we have
\begin{equation*}
    \Loss(\param + (\alpha_i\bfe_i +\alpha_j\bfe_j)) \leq \Loss(\param) + \gradL(\param)^{\ttT}(\alpha_i\bfe_i+\alpha_j\bfe_j) + \frac{\mathsf{L}_2}{2}\|\alpha_i\bfe_i+\alpha_j\bfe_j\|^2,
\end{equation*}
for some $\mathsf{L}_2 > 0$, where $\bfe_{i}$ and $\bfe_{j}$ represent unit vectors corresponding to $i^{\textth}$ and $j^\textth$ coordinates, respectively.
In this paper, $\|\cdot\|$ denotes the $\ell_2$ (Euclidean) norm for vectors and the spectral norm for matrices unless otherwise specified. }
 
\add{\noindent \textit{Assumption 2.} Let $\paramstar \in \RR^c$ be the global minimum of the loss function $\Loss$. For all $\param \in \RR^c$, it holds that
\begin{equation*}
    \frac{1}{2}\gradL(\param)^\ttT\left(\frac{1}{(c-1)}\mqfi(\param) + 2\beta\II\right)^{-1}\!\gradL(\param) \geq \mubar (\Loss(\param) - \Loss(\param^*)),\quad \text{for some 
 } \beta,\mubar>0.
\end{equation*}
The second assumption can be considered a regularized version of QGI inequality. The QGI inequality condition is independent of convexity, and it can be satisfied by non-convex functions with multiple saddle points. The QGI inequality is a slightly more general condition than the PL inequality. 
Below, we state the 2-QNSCD convergence theorem, where we refer to these assumptions in the proof provided in Supplementary Note 13. The proof mainly relies on the structure of estimators $\bar{\bfZ}$ and $\bfg$ and uses inequalities such as the Kiefer inequality \cite{kiefer1959optimum} and the operator Jensen's inequality \cite{nordstrom2011convexity}.}

\begin{theorem}[Convergence of 2-QNSCD]\label{thm:2QNCDconv} \add{Under the assumptions mentioned above, 2-QNSCD with a fixed learning rate $\eta = \beta/\mathsf{L}_2 
$, number of parameters $\nparams>2$, and the update rule \eqref{eq:2QNCD update} achieves an expected exponential convergence rate up to a residual error bounded by ${(\alpha^2\beta}/{4\mubar)}$. In particular, for every iteration $\ttt > 0$, 
\begin{enumerate}
    \item The expected difference between the loss at iteration 
$\ttt$ and the optimal loss is bounded by
\begin{equation}\label{eq:2QNCDconvergence_rate}
     \EE[\Loss(\param^{(\ttt)}) - \Loss(\paramstar)]
\leq \bigg(1-\frac{4\mubar\beta}{\nparams^2\mathsf{L}_2}\bigg)^{\ttt}(\Loss(\param^{(0)}) - \Loss(\paramstar)) +\frac{\alpha^2\beta}{4\mubar}
\end{equation}
\text{for some } $\beta >0$  and
$\alpha^2 =\max_{(i,j)} \EE\big[\||\Ztildeij^{\ttt}- ({2\beta}/{\nparams})\II_2|^{-1}\|^2\big]\EE\big[\|[\git,\gjt]\|^2\big]$.
\item As a consequence, for sufficiently large $\ttt$, 2-QNSCD converges within a residual neighborhood of the optimal solution, with the size given as
\begin{equation}
    \limsup_{\ttt\rightarrow\infty} \EE[\Loss(\param^{(\ttt)}) - \Loss(\paramstar)] \leq \frac{\alpha^2\beta}{4\mubar}. 
\end{equation}
\end{enumerate}}
\end{theorem}
\add{Note that $r_0\deq (\Loss(\param^{(0)}) - \Loss(\paramstar))$ is the difference between initial and optimal loss, \(\sigma \deq ({4 \mubar \beta}/{\nparams^2 \mathsf{L}_2}) =  ({4 \mubar \beta}/{(dL)^2\mathsf{L}_2})\) represents the contraction factor, where $c=dL$ (see the subsection PQC architecture). Additionally, \(\Delta \deq ({\alpha^2 \beta}/{4 \mubar})\) represents the asymptotic error bound. It depends on the moments of the E-QFIM and gradient estimators, assumptions on the loss function, and the regularization constant.}


\add{\emph{Sample Complexity:} We now provide the overall sample complexity bound for the 2-QNSCD algorithm. 
\begin{theorem}[Sample Complexity of 2-QNSCD]\label{thm:2QNSCDsample-comp}
Consider the 2-QNSCD convergence rate given by \eqref{eq:2QNCDconvergence_rate}. For any $\epsilon>\Delta$, the number of samples required by the 2-QNSCD algorithm to converge to the $\epsilon$-residual neighborhood of the optimal solution 
is 
    \[O\!\left(\nparams^2\ln\!\left(\frac{r_0}{\epsilon - \Delta}\right)\right).\]
\end{theorem}}
It is important to note that the sample complexity does not grow exponentially with the number of qubits 
$d$; instead, it scales quadratically, since 
$c=dL$. 
However, there are a few limitations: 
$1.$ The asymptotic error bound is constant and does not diminish with the number of iterations. $2.$ The contraction factor decreases quadratically with \(d\). As a consequence, for large systems, 2-QNSCD can potentially converge slowly to its asymptotic error bound. $3.$ Assumption 2 in Theorem 3 presents a regularized version of the QGI inequality. For large systems, this condition effectively reduces to the PL inequality. In the supplementary material, we provide examples where the metric-based QGI inequality holds, but the PL inequality does not, particularly in scenarios involving multiple local minima or saddle points. The failure of PL inequality in such cases may impede the convergence of the 2QNSCD algorithm for large systems.

\add{Therefore, for large systems, Theorems \ref{thm:2QNCDconv} and \ref{thm:2QNSCDsample-comp} suggest considering a hybrid approach. This involves starting with the 2-QNSCD to reach within the $\epsilon$-residual neighborhood of the optimal solution using fewer quantum samples and then gradually increasing the number of samples per iteration to reach an optimal solution.} 



\noindent \textbf{Numerical Results.}
To illustrate the utility of 2-QNSCD, we numerically assess the performance of the 2-QNSCD optimization method to train our QML model. In our experiments, we focus on the binary classification of quantum states with labels $\set{+1,-1}$ and conventional 0-1 loss to measure the predictor's accuracy. We use a synthetic dataset that generalizes the well-known dataset used in \cite{Mohseni2004} to multi-qubit systems. We provide the dataset generation process in the Methods section. 

\noindent \textit{Demonstration of Training Progress}:  
To demonstrate the progress during the learning phase, we group samples into multiple batches of size N$=600$. 
During each training step, we use a different batch, and the PQC processes every training sample of that batch and updates the model parameters accordingly. To investigate the speed of convergence, after each step, we compute the empirical loss of the PQC with the updated parameters given as 
\begin{equation}
{\Loss}_{\texttt{Emp}}(\param^{(\tti)}) = \frac{1}{\text{N}}\sum_{j=1}^\text{N}\mathds{1}_{\set{\yhat_j \neq y_j}}(\param^{(\tti)}),
\end{equation}
where $\yhat_j, y_j$ are the predicted and true labels, respectively, corresponding to the $j^{\textth}$ sample of the batch, $\param^{(\tti)}$ is the parameters at the end of $\tti^\textth$ step, and $\mathds{1}_{\set{\cdot}}$ denotes indicator random variable. In addition, we compute the average of per-sample expected loss 
(as described in \eqref{eqn:avgpersamplexpectedloss}) for all samples in the batch. We then plot the progress of training empirical loss and average per-sample expected loss against the number of steps. We provide additional details on the experiment setup and PQC architecture for different qubit configurations considered in the experiment in Supplementary Note 14. 

\noindent \textit{Comparison with different optimization schemes}: 
We compare the 2-QNSCD performance with the corresponding single-shot stochastic gradient descent, namely, RQSGD (randomized quantum stochastic gradient descent) \cite{heidari2022toward}. With N $= 600$, the 2-QNSCD performs $100$ parameter updates or iterations at each step. Therefore, to ensure a fair comparison, we provide 600 samples to RQSGD at each step to perform 100 iterations, but in two different ways. In the first method (2-RQSGD), six samples are used in each iteration to update only two parameters. 
In the second method (6-RQSGD), six samples are used in each iteration to update six parameters. Further details are provided in Supplementary Note 14.
Eventually, we compare the performance with the optimal expected loss within each batch. The closed-form expression of the optimal expected loss is given as
\begin{equation}\label{eqn:optimalLoss}
    \text{Optimal Expected Loss } \Loss
_{\texttt{Opt}} =  \frac{1}{2}\Big(1-\Big\|\frac{1}{\text{N}}\sum_{j=1}^\text{N}y_j\Phi_{x_j}\Big\|_1\Big),
\end{equation}
where $\set{|(\phi_{x_j}\>),y_j}_{j=1}^N$ is the labeled quantum samples in a batch of size N and $\|\cdot\|_1$ is the trace norm. The optimal expected loss for the binary quantum state classification problem is derived in (Lemma 3)\cite{heidari2022toward} using the Holevo-Helstrom theorem \cite{Holevo2012}. 
\begin{figure}
     \centering 
          \centering     
       \includegraphics[height=8.55in,width =\textwidth]{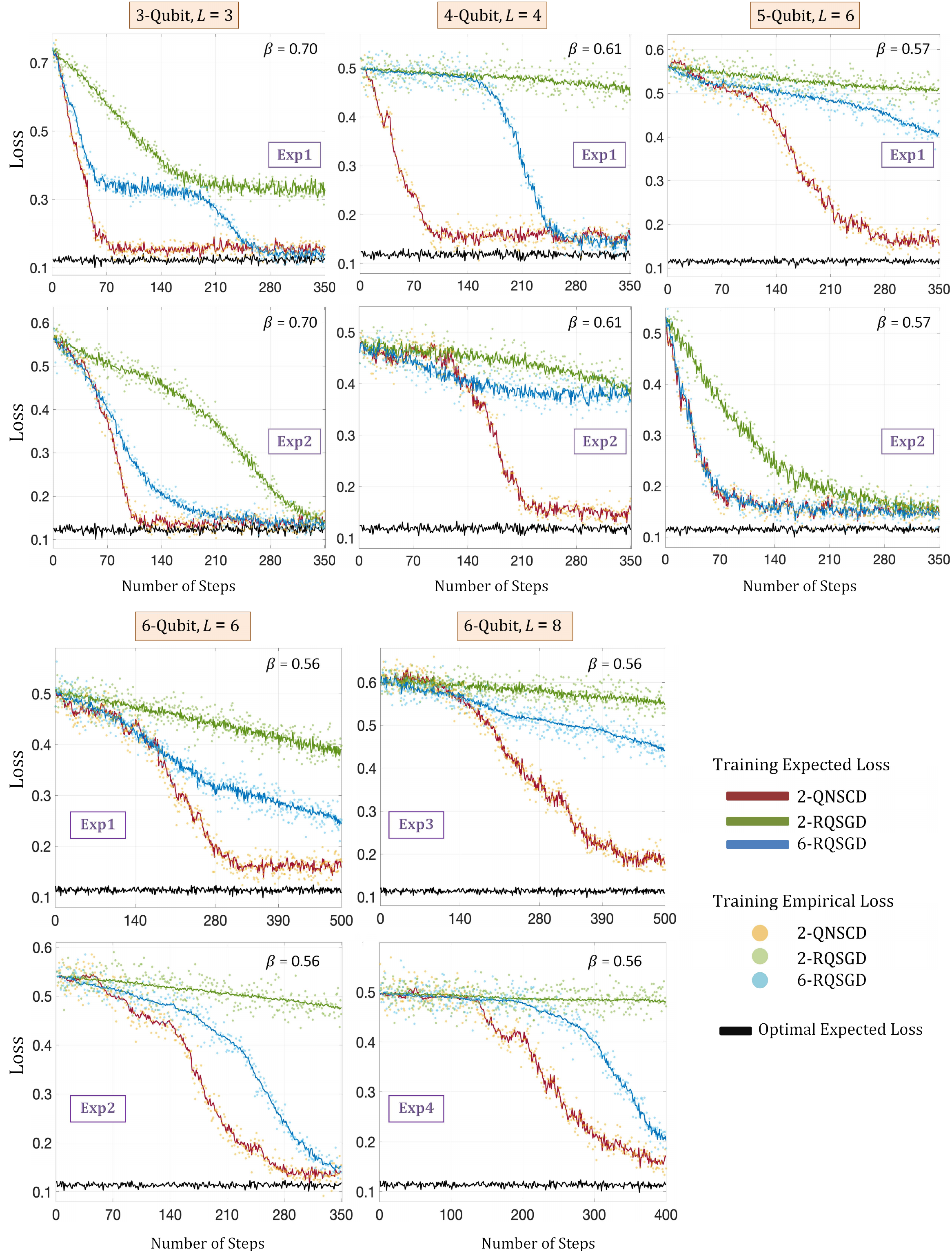}
        \caption{Performance of 2-QNSCD for 3, 4, 5, and 6 qubits. The results demonstrate that 2-QNSCD achieves faster convergence and demonstrates the ability to avoid saddle points and/or local minima.}
\label{fig:plot2QNSCDvs2QRSGD}
\end{figure}
Figure  \ref{fig:plot2QNSCDvs2QRSGD} illustrates the performance comparison between 2-QNSCD, 2-RQSGD, and 6-RQSGD. Notably, 2-QNSCD converges nearly to the optimal loss faster than 2-RQSGD and 6-RQSGD. For 3Q (3-Qubit) and 4Q (4-Qubit) cases, Exp-1 and Exp-2 correspond to two different initial points. In 3Q Exp-1, 2-RQSGD probably gets trapped inside a local minimum or around a saddle point, whereas 6-RQSGD converges to the optimal loss. However, for nearly $100$ steps, 6-RQSGD is stuck in a flat region of the loss function where the gradient is almost zero, whereas 2-QNSCD easily converges to optimal loss in just $70$ steps. In 3Q Exp-2, both RQSGD and 2-QNSCD converge to the optimal loss, but 6-RQSGD and 2-RQSGD require roughly $100 \times 600 = 6\times10^4$ and $14 \times10^4$ 
  more samples than 2-QNSCD, respectively, to converge to the optimal loss.
  
  In 4Q Exp-1, 2-RQSGD and 6-RQSGD are affected by the region, where the gradient is almost zero, for nearly $140$ steps. After this, 6-RQSGD drastically converges to the optimal loss, using roughly 
$9\times10^4$
  additional samples compared to 2-QNSCD, whereas 2-RQSGD fails to cross $60\%$ accuracy even after using 
$2.1\times10^5$
  samples (i.e., $350$ steps). In 4Q Exp-2, 2-QNSCD quickly converges to the optimal loss, whereas 2-RQSGD barely crosses $60\%$ accuracy after $350$ steps. Although 6-RQSGD initially converges faster than 2-QNSCD and 2-RQSGD, it likely gets trapped around saddle points or local minima.
  
For the 5Q case, Exp-1 and Exp-2 correspond to 5Q PQC-1 and 5Q PQC-2, respectively, with different initial points. In 5Q Exp-1, 2-QNSCD shows faster convergence, while in Exp-2, the performance of 6-RQSGD is similar to that of 2-QNSCD. However, it is important to note that 6-RQSGD has higher gate complexity than 2-QNSCD. This is because the gate complexity for E-QFIM estimation is lower than that for gradient estimation. The gate complexity for E-QFIM is, at most, the depth of the PQC, whereas the gradient estimation circuit uses an additional ancilla qubit and a controlled unitary operator, increasing the gate complexity beyond that of the PQC (for more details, see Methods). For the 6Q case, Exp-1 and Exp-2 use 6Q PQC-1, while Exp-3 and Exp-4 use 6Q PQC-2, which has eight layers, all with different initial points. These experiments also demonstrate the advantage of 2-QNSCD over RQSGD optimization methods. 

Furthermore, we generated a validation set of $1000$ samples to test and compare the accuracy of the training methods. The results are summarized in Table \ref{table:accuracy}. We note that the accuracy of 2-QNSCD is significantly better than 2-RQSGD and 6-RQSGD, and it is also close to the optimal value computed using \eqref{eqn:optimalLoss}.

\begin{table}[!htb]
\centering
\renewcommand{\arraystretch}{1.3} 
\setlength{\tabcolsep}{12pt} 
\begin{tabular}{|c|c|c|c|c|c|c|}
\hline
\multirow{11}{*}{\rotatebox[origin=c]{90}{\textbf{Accuracy}}} & \multicolumn{2}{|c|}{\textbf{No. of Qubits}} & \textbf{2-QNSCD} & \textbf{2-RQSGD} & \textbf{6-RQSGD} & \textbf{Optimal} \\ \cline{2-7}
 & \multirow{2}{*}{3Q} & Exp-1 & 84.6 $\pm$ 0.8\% & 65.9 $\pm$ 0.9\% & 85.8 $\pm$ 0.8\% & \multirow{2}{*}{87.3\%} \\ \cline{3-6}
&  & Exp-2 & 85.9 $\pm$ 1.1\% & 85.4 $\pm$ 0.8\% & 85.6 $\pm$ 0.9\% &  \\ \cline{2-7}
 & \multirow{2}{*}{4Q} & Exp-1 & 84.6 $\pm$ 1.2\% & 54.7 $\pm$ 0.7\% & 85.1 $\pm$ 0.7\% & \multirow{2}{*}{87.6\%} \\ \cline{3-6}
  &  & Exp-2 & 84.8 $\pm$ 0.6\% & 60.8 $\pm$ 1.5\% & 61.7 $\pm$ 1.3\% &  \\ \cline{2-7}
   & \multirow{2}{*}{5Q} & Exp-1 & 81.8 $\pm$ 0.8\% & 49.9 $\pm$ 1.3\% & 58.7 $\pm$ 1.7\% & \multirow{2}{*}{87.7\%} \\ \cline{3-6}
  &  & Exp-2 & 84.0 $\pm$ 1.0\% & 83.6 $\pm$ 1.1\% & 84.3 $\pm$ 1.1\% &  \\ \cline{2-7}
   & \multirow{4}{*}{6Q} & Exp-1 & 83.1 $\pm$ 1.0\% & 61.9 $\pm$ 1.7\% & 75.6 $\pm$ 1.1\% & \multirow{4}{*}{88\%} \\ \cline{3-6}
  &   & Exp-2 & 85.6 $\pm$ 1.2\% & 52.8 $\pm$ 1.3\% & 84.3 $\pm$ 1.4\% &  \\ \cline{3-6}
  &  & Exp-3 & 80.7 $\pm$ 1.1\% & 44.7 $\pm$ 1.8\% & 56.1 $\pm$ 1.4\% &  \\ \cline{3-6}
  &  & Exp-4 & 81.4 $\pm$ 1.0\% & 52.4 $\pm$ 2.2\% & 78.5 $\pm$ 1.1\% &  \\ \hline
\end{tabular}
\vspace{5pt}
\caption{Accuracy comparison of 2-QNSCD, 2-RQSGD, and 6-RQSGD for different qubit configurations and experimental setups. We note that the accuracy of 2-QNSCD is significantly better than 2-RQSGD and 6-RQSGD, and it is also close to the optimal value computed using \eqref{eqn:optimalLoss}.}
\label{table:accuracy}
\end{table}

\section{Discussion}
In this work, we develop a novel ensemble-based quantum information metric, E-QFIM, and design a quantum circuit capable of constructing an unbiased estimator of E-QFIM using only a constant number of samples and single-shot measurements. To reduce the quadratic computational cost to a constant computational cost per iteration, we ensure the estimator has a sparse structure. {
Furthermore, we provide an asymptotic error bound of the 2-QNSCD optimization method. 2-QNSCD exploits the underlying structure of state space to accelerate convergence and enhance the overall performance of VQAs, particularly in cases where stochastic gradient descent fails to achieve convergence.} Another significant advantage of 2-QNSCD is its data efficiency, requiring only a constant number of samples per iteration. This is a critical step forward, particularly in scenarios where obtaining a large number of quantum samples is not feasible. For example, in quantum sensing applications, quantum data is often scarce and costly to obtain. 2-QNSCD's data efficiency makes it an excellent option in such cases. These advancements mark a significant step toward developing a single-shot and metric-dependent optimization method for learning from quantum data.

We introduce 2-QNSCD within the framework of learning from quantum data, but its applicability can be extended to other variational algorithms, including VQE and QAOA. In the context of noisy quantum circuits, QNGD typically requires using a quantum information metric tensor on the space of density matrices, such as the Bures metric tensor. As discussed in the Introduction and Supplementary Note 1, estimating the Bures metric tensor has an exponential sample complexity. However, 2-QNSCD offers a solution to these challenges by treating the output of noisy quantum circuits as an ensemble of pure states and utilizing E-QFIM as a more efficient alternative to the Bures metric tensor.

Additionally, our current framework considers a tensor product of single-qubit Pauli rotations, but Algorithm \ref{alg:qfimEst} for constructing the unbiased estimator of E-QFIM can be easily extended to handle tensor products of multi-qubit Pauli rotations. This framework can also be extended to a more general PQC of the form $U(\param)  = \exp\{i\sum_{j} \param_j \bar{\sigma}_j\}$, where $\bar{\sigma}_j$ is a $d$-qubit Pauli strings, by designing the estimators using Suzuki-Trotter product formula \cite{hatano2005finding}. This flexibility enhances the robustness and adaptability of our approach for a wide range of applications.  Furthermore, we define E-QFIM for ensembles of pure states, but it can be extended to ensembles of density matrices. In Theorem \ref{thm:EQFI}, the density matrix $\rho$ (as a convex sum of pure states) is used. However, it can be replaced by the convex sum of density matrices from the ensemble, and a similar technique can be applied to construct an unbiased estimator of E-QFIM for an ensemble of density matrices. \add{This extension has significant applications in various quantum learning tasks involving real quantum data, such as entanglement detection and the quantum phase classification task.} An intriguing avenue of research involves exploring the potential of E-QFIM in quantum sensing or quantum communications applications. For instance, single-shot estimators can be used in designing quantum receivers or learning from entangled sensor networks. 

\section{Methods}
In this section, we present a series of quantum circuits to estimate $\mqfi$ and $\gradL$ independently and separately using only six quantum samples: four samples for $\mqfi$ and two samples for the gradient. 

\noindent \textbf{Estimation of the E-QFIM.}
For constructing an unbiased estimator of $\mqfi$, we start by randomly selecting a pair of coordinates $(\ap,\bq)$ corresponding to the parameters $\param_{(\a,\p)}\eqand \param_{(\b,\q)}$, respectively. Then, consider a $2\times 2$ sub-matrix of $\mqfi$ corresponding to the parameters $\param_{(\a,\p)}$ and $\param_{(\b,\q)}$, which we denote as $\Fab$. Thus, from \eqref{eqn:QNN_EQFI} it follows:
\begin{equation}\label{eq:2x2_EQFI}
    \Fab = \begin{bmatrix}
       0.25- [\Tr\parens{{\Upsilon_{\ap}}\rho}]^2 &  \Cov(\Upsilon_{\ap},\Upsilon_{\bq})_{\rho}\\ 
       \Cov({\Upsilon_{\ap}},\Upsilon_{\bq})_{\rho} & 0.25- [\Tr\parens{\Upsilon_{\bq}\rho}]^2
    \end{bmatrix}.
\end{equation}  
We first construct a $\nparams\times\nparams$ random matrix $\bfZ(\ap,\bq)$ such that 
its submatrix $\Zab$ is an unbiased estimate of $\Fab$, and 
all other entries of $\bfZ$ are zero. 
Subsequently, through some post-processing steps, we derive an unbiased estimate of $\mqfi$, denoted as $\Zbarab$. Now, the construction of $\Zab$ involves designing estimators for $\Tr(\Upsilon_{\ap}\rho)$, $\Tr(\Upsilon_{\bq}\rho)$, $[\Tr(\Upsilon_{\ap}\rho)]^2$, $[\Tr(\Upsilon_{\bq}\rho)]^2$, and $\Tr(\set{\Upsilon_{\ap},\Upsilon_{\bq}}\rho)$.
To understand this process, without loss of generality, we assume $\a \leq \b$. Let $z_{11}, z_{12}, z_{21}, z_{22}$ be the elements of the matrix $\Zab$. In the following, we use four samples $\set{\ketphi{x_i}}_{i\in[4]}$ to construct these elements. We provide a method to estimate the entries of the matrix as follows:

\noindent \textit{Estimator $z_{22}$}: We independently apply the first $(\b\!-\!1)$ layers of the PQC on the first two samples $\ketphi{x_1} \eqand \ketphi{x_2}$. Then, measure the $\q^\textth$ qubit along the basis of $\sigmabq$. Let $(v_1,v_2) \in \set{-1,+1}^2$ be the outcomes corresponding to the first and second measurements, respectively. Then, we compute $z_{22} \deq {0.25\cdot(1-v_1v_2)}$. 

For the off-diagonal terms, $z_{12} = z_{21}$  (due to the symmetry of $\mqfi$),
we provide a sequential measurement strategy to estimate the expectation of the anti-commutator, given below. Moreover, using this sequential strategy, we can construct the estimator $z_{11}$ as a bonus.

\noindent \textit{Estimators $z_{12},z_{11}$}: Using the above sequential strategy, for each of the remaining samples $\ketphi{x_3} \eqand \ketphi{x_4}$, we independently apply the first $(\a\!-\!1)$ layers of the PQC and measure the $\p^\textth$ qubit along the eigenvectors of $\sigmaap$.
The collapsed (post-measured) state is then passed through layer $\a$ up to layer $(\b\!-\!1)$ of the PQC, and the $\q^\textth$ qubit of the output state is measured along the eigenvectors of  $\sigmabq$.   
Let $(u_1,w_1), (u_2,w_2) \in\{-1,+1\}$ be the first and second measurement outcomes for $\ketphi{x_3} \eqand \ketphi{x_4}$, respectively. Then, we compute off-diagonal entries as follows:
    $z_{21} = z_{12} \deq 0.125\cdot(u_1w_1+u_2w_2) - 0.0625\cdot(u_1+u_2)(v_1+v_2)$. Finally, we compute the diagonal term $z_{11} \deq  {0.25\cdot(1-u_1u_2)}$. Now, combining these estimators, we obtain the matrix $\Zab$.
    


 \noindent \textit{Estimator of $\mqfi$}: After constructing the submatrix $\Zab$,  we now aim to construct a non-singular and unbiased estimator of the entire E-QFIM matrix. To achieve this, we introduce a positive parameter $\beta > 0$, which serves as a regularization constant. Toward this, scale the  diagonal terms of $\Zab$ by $\frac{1}{(\nparams-1)}$ and add $\beta$
 to it. This yields a $2 \times 2$ matrix, denoted as $\tilde{\bfZ}_{[\ap,\bq]}$. 
The remaining $\nparams\times\nparams$ matrix $\tilde{\bfZ}$ is completed with zeros as other entries. 
 Finally, we construct the estimator of $\mqfi$ as \begin{equation}\label{eqn:QFIMestimatorZbar}
     \bar{\bfZ}(\ap,\bq) \deq \frac{\nparams(\nparams-1)}{2}\Big(\Tilde{\bfZ} - \frac{2\beta}{\nparams}\II\Big).
 \end{equation}
 Here, $\II$ denotes $\nparams\times \nparams$ identity matrix. The parameter $\beta$ is used to ensure the positive definiteness of $\bar{\bfZ}_{[\ap,\bq]}$, which primarily contributes to the update rule of 2-QNSCD. It prevents the condition number of the random matrix $\bar{\bfZ}_{[\ap,\bq]}$ from scaling with the number of model parameters and thus mitigates any numerical instability. The final matrix  $\bar{\bfZ}{(\ap,\bq)}$ obtained after the classical post-processing of $\Zab$ is an unbiased estimator of $\mqfi$, as stated in Theorem \ref{thm:unbiasedQFIM}. Moreover, it can be easily seen that the gate complexity for constructing the unbiased estimator of $\mqfi$ is, at most, the gate complexity of $U(\param)$. This procedure is summarized in Algorithm \ref{alg:qfimEst}. 

\begin{algorithm}[ht]
\caption{\textit{One-Shot} E-QFIM Estimator}
\label{alg:qfimEst}
\DontPrintSemicolon
\setstretch{1.25}
\LinesNumbered

  \SetKwFunction{FMain}{Main}
  \SetKwFunction{FSum}{E-QFIM\_Estimator}
  \SetKwFunction{FSub}{Sub}
 
  \SetKwProg{Fn}{Function}{:}{}
  \Fn{\FSum{$\set{\ketphi{x_i}}_{i\in[4]},\ap,\bq,\beta$}}{
  \%\% $\{\ketphi{x_i}\}_{i\in[4]}:$ i.i.d. quantum samples\;
  \%\% $(\ap,\bq):$ a random pair of distinct coordinates (assume $\a \leq \b$)\;
  \%\% $\beta > 0:$ regularization constant\;
  \tcc{Consider $\ketphi{x_1} \eqand \ketphi{x_2}$}
  \For{$i=1$ to $2$}{
    Apply the first $(\b\!-\!1)$ layers of $U(\param)$ on $\ketphi{x_i}$\;
    Measure $\q^\textth$ qubit along the eigenvectors of $\sigmabq$\;
    $v_i \in \set{-1,+1} \leftarrow$ measurement outcome
    }

    \tcc{Consider $\ketphi{x_3} \eqand \ketphi{x_4}$}
  \For{$i=3$ to $4$}{
    Apply the first $(\a\!-\!1)$ layers of $U(\param)$ on $\ketphi{x_i}$\;
    Measure $\p^\textth$ qubit along the eigenvectors of $\sigmaap$\;
    $u_{i-2} \in \set{-1,+1} \leftarrow$ measurement outcome\;
    $\ket{\phi_{x_i}^\a} \leftarrow$ post-measured state\;
    Apply layers from $\a$ to $(\b\!-\!1)$ of $U(\param)$ on $\ket{\phi_{x_i}^\a}$\;
    Measure $\q^\textth$ qubit along the eigenvectors of $\sigmabq$\;
    $w_{i-2} \in \set{-1,+1} \leftarrow$ measurement outcome\;
    }
\tcc{Construct E-QFIM estimator}    
$\tilde{\bfZ} \leftarrow 0_{\nparams \times \nparams}$ (all-zero matrix)\;
        $\Tilde{\bfZ}_{[\ap,\bq]} = \begin{bmatrix}
            \frac{(1-u_1u_2)}{4\cdot(\nparams-1)} +\beta & \frac{(u_1w_1+u_2w_2)}{8} - \frac{1}{4}\frac{(u_1+u_2)}{2}\frac{(v_1+v_2)}{2}\\
            \frac{(u_1w_1+u_2w_2)}{8} - \frac{1}{4}\frac{(u_1+u_2)}{2}\frac{(v_1+v_2)}{2} &
            \frac{(1-v_1v_2)}{4\cdot(\nparams-1)}+\beta
        \end{bmatrix}$\;
        \vspace{2pt}
\KwRet $\bar{\bfZ} \deq \frac{\nparams(\nparams-1)}{2}(\tilde{\bfZ}-(\frac{2\beta}{\nparams})\II)$
}
\end{algorithm}
This completes the construction of an unbiased estimator of the E-QFIM. Next, we design an unbiased estimator of the gradient of the loss function.

\noindent \textbf{Estimation of the Gradient.}
Consider a pair of coordinates $(\ap,\bq)$ corresponding to the parameters $\param_{(\a,\p)} \eqand \param_{(\b,\q)}$.
We use two samples along with their true labels $(\ketphi{x_1},y_1)\eqand(\ketphi{x_2},y_2)$ to construct unbiased estimators $g_{\ap} \eqand g_{\bq}$ of the elements of $\gradL$ corresponding to these parameters, respectively. The method is as follows:

\noindent \textit{Estimator $g_{\ap}$}: We apply the first $(\a\!-\!1)$ layers of the PQC and add an ancilla qubit $|+\>$ to the output quantum state $|\phi_{x_1}^\a\>$. This creates the state $\Tilde{\Phi}_{x_1}^\a \deq (W_{(\a:1]} {\Phi}_{x_1} W_{[1:\a)}^{\dagger}) \tensor |+\>\<+|$. Then, we apply the following unitary matrix $\bfV \deq e^{i\pi \Sigmaap/4} \tensor |0\>\<0| + e^{-i\pi \Sigmaap/4} \tensor |1\>\<1|$ on $\Tilde{\Phi}_{x_1}^\a$. Next, we apply the remaining 
layers of the PQC 
on the state $\bfV\Tilde{\Phi}_{x_1}^\a\bfV^{\dagger}$, and  measure the state by the quantum measurement $\tilde{\Lambda} \deq\set{\Lambda_{\yhat}\tensor \ketbra{\ttb}: \yhat \in \calY, \ttb \in \set{0,1}}$. Let $(\yhat,\ttb)$ be the outcome of the measurement. Finally, we compute $g_{\ap} \deq (-1)^{(1+\ttb)}\ell(y_1,\yhat)$. 

\noindent \textit{Estimator $g_{\bq}$}: Following a similar procedure as for $g_{\ap}$, we construct $g_{\bq}$ with the correspondence $\ap \leftrightarrow \bq$ and using the pair $(\ketphi{x_2},y_2)$. Note that, using Lemma \ref{lem:obs_commutator}, one can easily show 
\begin{equation}\label{eqn:1cGrad_est}
    \EE[g_{\ap}|\ketphi{x_1},y_1] = \gradL_{(\a,\p)}(\param,\ketphi{x_1},y_1) \eqand \EE[g_{\bq}|\ketphi{x_2},y_2] = \gradL_{(\b,\q)}(\param,\ketphi{x_2},y_2),
\end{equation}
where the expectation is over the measurement outcome and the input quantum state.

\noindent \textit{Estimator of $\gradL$}: After constructing $g_{\ap} \eqand g_{\bq}$, we construct the estimator of $\gradL$ as
\begin{equation}\label{eqn:gradEst}
     \bfg{(\ap,\bq)} \deq \left(\frac{c}{2}\right)(g_{\ap} \bfe_{\ap} + g_{\bq} \bfe_{\bq}), 
\end{equation}
where $\bfe_{\ap}$ and $\bfe_{\bq}$ represent unit vectors corresponding to $\ap$ and $\bq$, respectively. 
This gradient estimation procedure is summarized in Algorithm \ref{alg:gradientEst}. 

\begin{algorithm}[ht]
\caption{\textit{One-Shot} Gradient Estimator}
\label{alg:gradientEst}
\DontPrintSemicolon
\setstretch{1.25}
\LinesNumbered

  \SetKwFunction{FSum}{Gradient\_Estimator}
  
  \SetKwProg{Fn}{Function}{:}{}
  \Fn{\FSum{$\{(\ketphi{x_i},y_i)\}_{i\in[2]},\ap,\bq$}}{
  \%\% $\{(\ketphi{x_i},y_i)\}_{i\in[2]}:$ i.i.d. quantum samples with true labels\;
  \%\% $(\ap,\bq):$ a random pair of distinct coordinates (assume $\a \leq \b$)\;
  Let $\mathsf{l}_1 \leftarrow \a \eqand \mathsf{l}_2 \leftarrow \b$\; 
  \tcc{Consider $\ketphi{x_1}\eqand \ketphi{x_2}$}
    \For{$i=1$ to $2$}{
    Apply the first $(\mathsf{l}_i\!-\!1)$ layers of $U(\param)$ on $\ketphi{x_i}$\;
    Add an ancilla qubit $|+\>$ to output of layer $(\mathsf{l}_i\!-\!1)$\;
    $({\Phi}_{x_i}^{\mathsf{l}_i} \tensor |+\>\<+|)\leftarrow$ output state\;
    Apply $\bfV \deq e^{i\pi \Sigmaap/4} \tensor |0\>\<0| + e^{-i\pi \Sigmaap/4} \tensor |1\>\<1|$ on $({\Phi}_{x_i}^{\mathsf{l}_i} \tensor |+\>\<+|)$\;
    Apply layers from $\mathsf{l}_i$  to $L$ on $\bfV ({\Phi}_{x_i}^{\mathsf{l}_i} \tensor |+\>\<+|)\bfV^\dagger$\;  
    Measure the resulting state with $\set{\Lambda_{\yhat} \tensor |\ttb\>\<\ttb|: \yhat\in \calY, \ttb\in \set{0,1}}$\;
    $(\yhat_i,\ttb_i) \leftarrow$ measurement outcomes
}
$g_{\ap} \leftarrow (-1)^{(1+\ttb_1)}\ell(\yhat_1,y_1)$\;
$g_{\bq} \leftarrow (-1)^{(1+\ttb_2)}\ell(\yhat_2,y_2)$\;
\KwRet $\bfg{(\ap,\bq)} \deq \left(\frac{c}{2}\right)(g_{\ap} \bfe_{\ap} + g_{\bq} \bfe_{\bq})$
}
\end{algorithm}

\noindent \textbf{2-QNSCD Update Rule.}
We construct the following update rule that updates only two randomly selected parameters (indexed by $\ap \eqand \bq$) at each iteration,
\begin{equation}\label{eqn:updaterule2QNSCDfinal}
  \param^{(\ttt+1)} = \param^{(\ttt)} - \eta_\ttt \big|\bar{\bfZ}^{\ttt}(\ap,\bq)\big|^{-1} \bfg^\ttt(\ap,\bq),
\end{equation}
where $\beta$ is chosen according to Remark \ref{rem:beta} as stated below. 
\begin{remark}\label{rem:beta}
     It is important to note that, given the structure of $\bar{\bfZ}$, by appropriately choosing $\beta$, the sub-matrix $\big(\tilde{\bfZ}_{[\ap,\bq]}-\big(\frac{2\beta}{\nparams}\big)\II_{2}\big)$ can be made positive definite with probability 1.
\end{remark}
Because of the structure of $\bar{\bfZ}\eqand \bfg$, \eqref{eqn:updaterule2QNSCDfinal} can be re-written as
\begin{equation}\label{eqn:2QNSCDfinalequivalent}[\param^{(\ttt+1)}_{(\a,\p)},\param^{(\ttt+1)}_{(\b,\q)}]^\ttT \leftarrow [\param^{(\ttt)}_{(\a,\p)},\param^{(\ttt)}_{(\b,\q)}]^\ttT - \eta_{\ttt} \bigg|\tilde{\bfZ}^{\ttt}_{[\ap,\bq]} - \bigg(\frac{2\beta}{\nparams}\bigg)\II_2\bigg|^{-1}[g^{\ttt}_{\ap}, g^{\ttt}_{\bq}]^\ttT,
\end{equation}
where $\II_2$ is the $2\times 2$ identity matrix and $\eta_\ttt$ contains the normalizing constant $(\nparams\!-\!1)$.
From the description of the 2-QNSCD algorithm in \eqref{eqn:2QNSCDfinalequivalent}, it is evident that only $\big(\tilde{\bfZ}_{[\ap,\bq]}-\big(\frac{2\beta}{\nparams}\big)\II_{2}\big)$ contributes to the update of parameters. 
Therefore, for a given $\nparams$, choosing an appropriate $\beta$ (Remark \ref{rem:beta}), makes the use of $|\bar{\bfZ}^\ttt|$ equivalent to that of $\bar{\bfZ}^\ttt$ in the update rule \eqref{eqn:updaterule2QNSCDfinal}. However, we keep the former for the sake of mathematical rigor in Algorithm \ref{alg:2-QNSCD}.

\noindent \textbf{Dataset Generation}: In the dataset generation process, first, consider the following three $d$-qubit quantum states:
       \[\ket{\phi_1(\bfu)} \deq 
    \sum_{j=0}^{(2^{d}-2)/2} \frac{\bfu_j}{\|\bfu\|}\ket{\texttt{bin}(2j)}, \
    \ket{\phi_2(\bfu)} \deq 
    \sum_{j=0}^{(2^{d}-2)/2} (-1)^{(j\!\texttt{ mod 
 }\!2+1)}\frac{\bfu_j}{\|\bfu\|}\ket{\texttt{bin}(2j+\mathds{1}_{\set{j\!\texttt{ mod }\!2 = 0}})}, \]
    \[\ket{\phi_3(\bfu)} \deq 
    \sum_{j=0}^{(2^{d}-2)/2} \frac{\bfu_j}{\|\bfu\|}\ket{\texttt{bin}(2j+\mathds{1}_{\set{j\!\texttt{ mod }\!2 =0}})},\]
where the vector $\bfu \in [0,1]^{2^{d-1}}$ and $\texttt{bin}$ is the function that converts a decimal number to its binary representation. To understand the structure of the above quantum states, consider the case where 
$d=3$. In this scenario, the quantum states are given as follows:
\begin{align*}
     \ket{\phi_1(\bfu)} &= (\bfu_0\ket{000} + \bfu_1\ket{010} +\bfu_2\ket{100}+\bfu_3\ket{110})/{\|\bfu\|},\\ 
        \ket{\phi_2(\bfu)} &= (-\bfu_0\ket{001} + \bfu_1\ket{010} -\bfu_2\ket{101} +\bfu_3\ket{110})/{\|\bfu\|},\\
        \ket{\phi_3(\bfu)} &= (\bfu_0\ket{001} + \bfu_1\ket{010} +\bfu_2\ket{101} +\bfu_3\ket{110})/{\|\bfu\|}.
\end{align*}
The quantum states to be classified are: $\ket{\phi_1(\bfu)}$ with label $y=+1$ and $\set{\ket{\phi_2(\bfu)},\ket{\phi_3(\bfu)}}$ with label $y=-1$. For each quantum sample, a state is generated from the set $\set{\ket{\phi_1(\bfu)},\ket{\phi_2(\bfu)},\ket{\phi_3(\bfu)}}$ with equal probability. This implies, $\prob{y=-1} = 2\cdot\prob{y=+1} = 2/3$. Additionally, for each sample, the vector $\bfu$ is selected randomly, independently, and with the uniform distribution on $[0,1]^{2^{d-1}}$. 
For better understanding, Fig. \ref{fig:DATASET} illustrates the dataset for the $2$-qubit case. 
\begin{figure}[!htb]
    \centering
\includegraphics[scale=0.13]{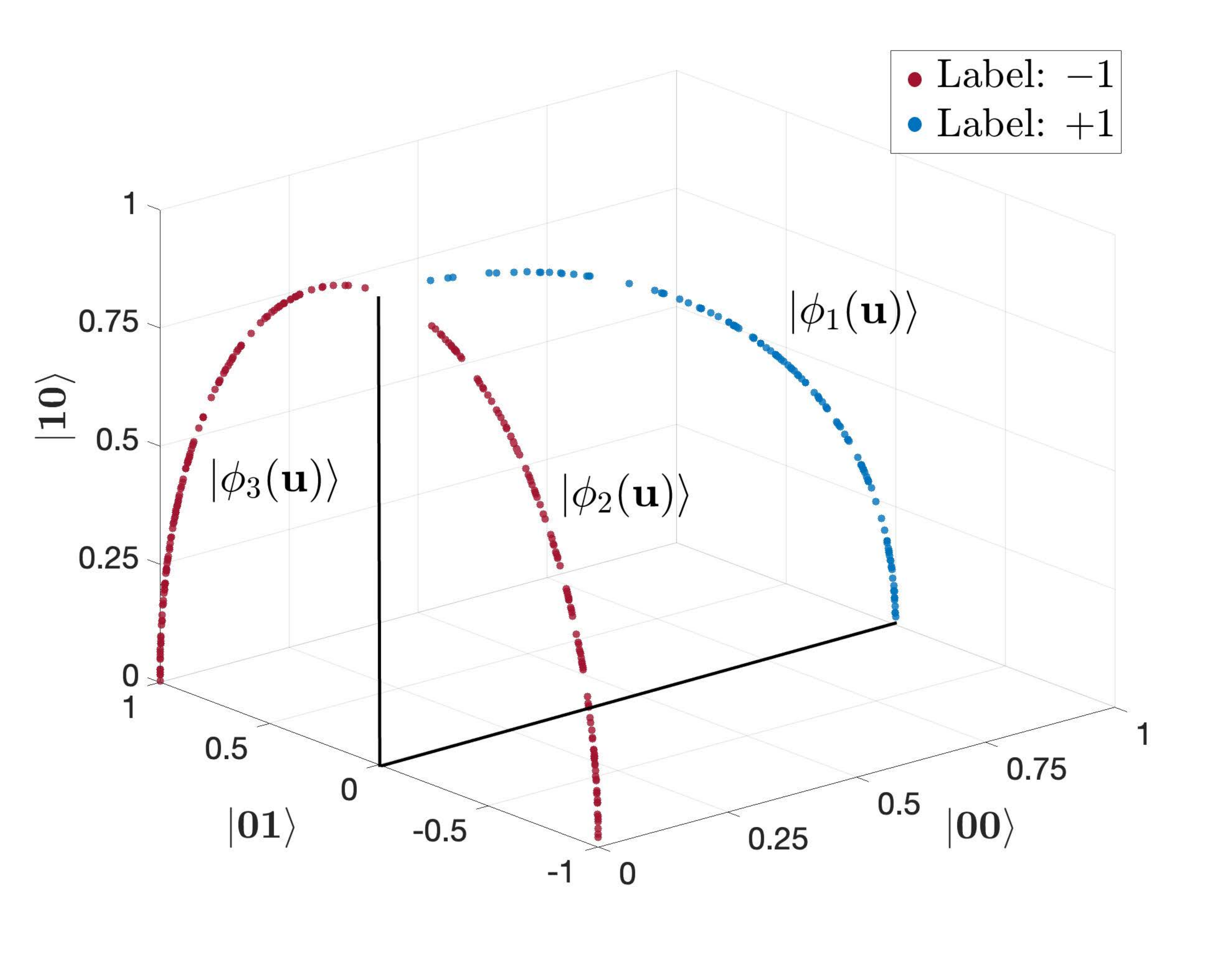}
    \caption{2-qubit synthetic datasets, where the quantum states are given as $\ket{\phi_1(\bfu)} = (\bfu_0\ket{00} + \bfu_1\ket{10})/\|\bfu\|, 
        \ket{\phi_2(\bfu)} = (-\bfu_0\ket{01} + \bfu_1\ket{10})/\|\bfu\|, \eqand 
        \ket{\phi_3(\bfu)} = (\bfu_0\ket{01} + \bfu_1\ket{10})/{\|\bfu\|}.$ The state $\ket{\phi_1(\bfu)}$ has label +1 and is colored in blue. The states $\ket{\phi_2(\bfu)}\eqand \ket{\phi_3(\bfu)}$ have label -1 and are colored in red.}
    \label{fig:DATASET}
\end{figure}

\noindent \textbf{{\large Data availability}} 

\noindent All data generated or analyzed during this study are included in this article.

\noindent \textbf{{\large Code availability}}

\noindent All the experimental results and source code implementations are available at \href{https://github.com/mdaamirQ/2-QNSCD}{https://github.com/mdaamirQ/2-QNSCD}.

\noindent \textbf{{\large Acknowledgments}} 

\noindent This work is supported in part by a gift from Accenture, the 2023 QUAD Fellowship, and NSF grant CCF-2211423.
We thank Touheed Anwar Atif  (with Los Alamos National Labs), and Hassan Naseri, Carl Dukatz, and Kung-Chuan Hsu (with Accenture)  for the initial discussions on this topic. \add{We thank the anonymous reviewers for their insightful comments and suggestions, which significantly enhanced the quality of the manuscript.}

\noindent \textbf{{\large Authors contributions}:} 

\noindent The project was conceived by M.A.S, M.H., and S.S.P. Experiments were conducted by M.A.S, and reviewed by M.H. and S.S.P. The manuscript was written by M.A.S., M.H., and S.S.P. All authors reviewed the final manuscript. M.H. and S.S.P. supervised the work and verified the main results. 

\noindent \textbf{{\large Competing Interests}}

\noindent The authors declare no competing financial or non-financial interests.

\bibliographystyle{ieeetr}
\bibliography{references}

\title{Supplement Material: Quantum Natural Stochastic Pairwise Coordinate Descent }
\maketitle

\renewcommand{\thesection}{\arabic{section}} 
\setcounter{section}{0}

\setcounter{page}{1}  
\renewcommand{\thepage}{A\arabic{page}}

\renewcommand{\theequation}{A.\arabic{equation}}
\setcounter{equation}{0} 

\renewcommand{\thefigure}{A\arabic{figure}} 
\setcounter{figure}{0} 

\renewcommand{\thetheorem}{A.\arabic{theorem}}
\setcounter{theorem}{0}

\renewcommand{\thelemma}{A.\arabic{lemma}}
\setcounter{lemma}{0}
\section{Background and Related Work}
\label{app:background}


Recent demonstrations, such as the parameter shift rule and finite differencing \cite{Mitarai2018,Schuld2018,Harrow2021}, showcase the ability to compute gradients for quantum circuits by estimating the loss function 
using multiple identical copies of a quantum state\footnote{Note that, unlike classical learning, in quantum domain exact computation of loss function is not possible because quantum states $\ketphi{x}$ are unknown, i.e., we do not have complete knowledge of $\ketphi{x}$. If we perform measurement on a quantum state to access information, it will result in a state collapse. Therefore, one can only estimate the loss function using multiple identical copies of a quantum state.}. Therefore, 
the \ac{SGD} method can be considered to train the quantum learning model via an update rule of the form:
\vspace{-5pt}
\begin{align}\label{eqn:GDideal}
\param^{(\ttt+1)} = \param^{(\ttt)} - \eta_\ttt \gradL(\param^{(\ttt)},\ket{\phi_{x_\ttt}}, y_\ttt),
\vspace{-8pt}
\end{align}
where $\eta_\ttt$ is the learning rate at iteration $t$. 
Further, to mitigate the need for multiple identical copies of a quantum state, \cite{heidari2022toward} 
proposed a single-shot gradient estimations method, namely, a randomized quantum stochastic gradient descent (RQSGD) optimization technique.
Nevertheless, as mentioned in the introduction, GD is inefficient at handling saddle points. These points are often surrounded by a plateau of small curvature, causing the gradient to diminish rapidly and significantly slowing down the training process.
To address this, one may consider a second-order optimization method, such as the Newton method \cite{boyd2004convex}, which utilizes the Hessian of the loss function. The update step of the Newton method replaces the gradient by multiplying the gradient with the inverse of the Hessian matrix. However, this does not address saddle points satisfactorily, and instead, saddle points become attractive under Newton dynamics, as argued in \cite[Section 4]{dauphin2014identifying}. 
In 1998, Amari \textit{et al.} \cite{amari1998natural} proposed the concept of natural gradient descent (NGD), where the Hessian is replaced with Fisher information. 
NGD has been shown to effectively address certain saddle point structures, as argued by  \cite{rattray1998natural,inoue2003line}. Furthermore, NGD is invariant under arbitrary smooth and invertible reparameterizations, whereas Newton method is invariant only under affine coordinate transformation \cite{amari1998natural,martens2020new}. 


In a similar spirit, QNGD \cite{stokes2020quantum} has been proposed as an optimization method based on the information geometry of the space of quantum states, which uses a quantum information (Reimannian) metric tensor as 
\begin{align}\label{eqn:QNGD}
\param^{(\ttt+1)} = \param^{(\ttt)} - \eta_\ttt \fubini(\param^{(\ttt)})^{-1} \gradL(\param^{(\ttt)},\ket{\phi_{x_\ttt}}, y_{\ttt}),
\end{align}
where $\fubini(\param^{(t)})$ denotes quantum information metric tensor. In \eqref{eqn:QNGD}, $\fubini^{-1}\gradL$ is the steepest direction in the space of quantum states. Essentially, here, each optimization step computes the steepest descent direction of the per-sample expected loss around the local value of \( \param^{(\ttt)} \) on the space of quantum states and updates $( \param^{(\ttt)} \rightarrow \param^{(\ttt+1)} )$ accordingly. A number of studies have demonstrated the performance gains achieved by QNGD over GD. For instance, \cite{wierichs2020avoiding} shows that QNGD provides an advantage for optimizing {parameterized} quantum systems by taking an optimal parameters path compared to other optimization strategies. For more details, we suggest \cite{stokes2020quantum,yamamoto2019natural,wang2023experimental,FQNGD,meyer2021fisher,kolotouros2023random}. 

In the literature, the quantum information metric tensor is derived using a suitably defined distance measure, denoted as $d(\param,\param^{'})$, in the space of quantum states. The squared infinitesimal distance can be expressed in terms of quantum information metric tensor (Taylor series around $\dparam = 0$) as:
\begin{equation*}
    \text{d}s^2 \deq d^2(\param,\param +\dparam) \approx  \sum_{i,j}\fubini(\param)_{(i,j)} \dparam_{i}\dparam_j,
\end{equation*}
where the first-order term goes to zero as $\dparam =0$ corresponds to a minimum\footnote{A distance measure $d(\param,\param^{'})$ is non-negative and equals zero for identical parameters, i.e., $d(\param, \param) = 0$ is a minimum.}, and the second-order term is the first non-zero contribution of the Taylor series expansion around $\dparam = 0$. In the case of pure states, the distance measure is defined using the fidelity between pure states as 
 \begin{equation}\label{eqn:fid_distance}
     d^2(|\phi(\param)\>,|\phi(\param^{'})\>) = \big(2-2\sqrt{f_\phi(\param,\param^{'})}\big),
 \end{equation}
 where $f_\phi(\param,\param^{'}) \deq |\<\phi(\param)|\phi(\param^{'})\>|^2$, and $\fubini$ reduces to the Fubini-Study metric tensor $\mathsf{F}^\phi$ (see Def.~\ref{def:fubinistudymetric} below), as demonstrated by Petz et al. \cite{petz1996geometries,petz1998information}.
 \vspace{-8pt}
 \begin{definition}
[Fubini-Study Metric Tensor]\label{def:fubinistudymetric}
The Fubini-Study metric tensor, denoted as $\mathsf{F}^\phi$, is a Riemannian metric tensor defined on the complex projective space $\mathbf{C}\mathbf{P}^n$, which is the space of pure states with global phase factored out. For a given parameterized pure state $|\phi(\param)\>$, the entries of $\mathsf{F}^\phi$ are given as: 
\vspace{-8pt}
\begin{equation*}
    \mathsf{F}_{(i,j)}^\phi(\param) = \Re\{\bra{\partial_i \phi(\param)}\ket{\partial_j \phi(\param)\> - \<{\partial_i \phi(\param)}\ketbra{\phi(\param)}\partial_j \phi(\param)}\},
    \vspace{-8pt}
\end{equation*} 
where $i,j\in [c]$ and $c$ is the number of parameters. It measures the effect of changing parameters on the underlying parameterized pure state.
\end{definition} 
 For density matrices, $\fubini$ is generalized to the Bures metric tensor $\mathsf{F}^\rho$ \cite{braunstein1994statistical,bures1969extension,liu2020quantum,liu2014fidelity} (see Def.~\ref{def:Buresmetric} below) using Bures distance \cite{nielsen2010quantum} between density matrices given as
 \begin{equation}\label{eqn:buresdistance}
     d_{\text{B}}^2(\rho(\param),\rho(\param^{'})) = \big(2-2\sqrt{f_\rho(\param,\param^{'})}\big),
 \end{equation} 
 where $f_\rho(\param,\param^{'}) \deq (\Tr\{({\sqrt{\rho(\param)}\rho(\param^{'})\sqrt{\rho(\param)}})^{1/2}\})^2$ is the Uhlmann fidelity \cite{wilde2013quantum} and $\rho(\param)$ is a parameterized density matrix. 

\vspace{-8pt}
\begin{definition}[Bures Metric Tensor]\label{def:Buresmetric}
    The Bures metric tensor, denoted as $\mathsf{F}^\rho$, is a Riemannian metric tensor defined on the space of density matrices. 
    For a given parameterized density matrix $\rho(\param)$, the entries of $\mathsf{F}^\rho$  are given as
    \vspace{-8pt}
\begin{align*}
    \mathsf{F}^\rho_{(i,j)} &\deq \sum_{k} \frac{1}{4}\frac{(\partial_i \lambda_k)(\partial_j \lambda_k)}{\lambda_k}  + \sum_{k} \lambda_k \Re{\<\partial_i \lambda_k|\partial_j \lambda_k\>} - \sum_{k_1,k_2} \frac{2 \lambda_{k_1}\lambda_{k_2}}{\lambda_{k_1}+\lambda_{k_2}}\Re{\<\partial_i \lambda_{k_1}|\lambda_{k_2}\>\<\lambda_{k_2}|\partial_j \lambda_{k_1}\>},
    \vspace{-8pt}
\end{align*}
where $i,j\in [c]$, $c$ is the number of parameters, and $\rho(\param) = \sum_k \lambda_k|\lambda_k\>\<\lambda_k|$ is the spectral decomposition of $\rho(\param)$ with parameterized eigenvalues $\lambda_k (\neq 0)$ and parameterized eigenvectors $|\lambda_k\>$. 
\end{definition}

Unfortunately, none of these information metric tensors can be efficiently used in QNGD methods for learning from quantum data due to the following reasons. The Fubini-Study metric tensor is defined only for pure states, {whereas} the feature set in the model is an ensemble of pure states.
Therefore, 
one may consider the Bures metric tensor,  defined on the density matrix of the feature set. 
However, computing the Bures metric tensor is much more involved and requires
a full tomography of the feature density matrix 
to access the eigenvalue and eigenvectors of the density matrix to compute $\mathsf{F}^\rho$,  requiring an exponential number of identical copies of $\rho$. Precisely, to measure a $d$-qubit quantum state $\rho$ of rank $r$ using quantum state tomography with $\varepsilon$-accuracy in terms of trace distance, requires $O(2^dr^2/\varepsilon^2)$ copies \cite{gross2010quantum,haah2016sample}. Alternatively, many approximation methods have been proposed, including a pure-state approximation of mixed-quantum state \cite{koczorQNGmixed,koczor2021dominant,koczor2021exponential} and truncated $\mathsf{F}^\rho$ \cite{sone2021generalized}. In these works, the authors approximated the Bures metric tensor to a few dominant eigenvectors of $\rho$ with an additional error, which decreases exponentially in the number of copies of $\rho$. For more information, refer to \cite{beckey2022variational,rath2021quantum,vitale2023estimation}. While these works offer various approaches, they share a crucial drawback: requiring an exponentially increasing number of copies of $\rho$ to achieve negligible error in metric estimation.


Furthermore, QNGD requires evaluation of the quantum information metric tensor at each iteration, which can be computationally intensive. This becomes particularly demanding when dealing with a large number of PQC parameters as computing $\fubini$ requires evaluating $O(\nparams^2)$ terms. 
Various techniques have been proposed to mitigate this computational cost, reducing it to linear complexity, such as block-diagonal approximation \cite{stokes2020quantum} and second-order simultaneous perturbation methods \cite{gacon2021simultaneous} for the Fubini-Study metric tensor. However, it is important to note that such approximations may not precisely capture parameter correlations. Therefore, QNGD can not perform well in situations where parameters are highly correlated. In addition, computing $\fubini$ (exact or approximation) involves calculating fidelity between two quantum states, which can be done using quantum circuits such as the SWAP test \cite{buhrman2001quantum}, the Hadamard test \cite{cleve1998quantum}, and the compute-uncompute method \cite{havlivcek2019supervised}. However, these techniques require an exponential number of measurements for the estimation of fidelities associated with the metric tensor. 
Precisely, pure state fidelity is approximated with up to $\varepsilon$ error by performing $\Theta(1/\varepsilon^2)$ measurements  \cite{wierichs2022general}, where $\varepsilon \in (0,1)$.

To address these challenges, we propose a novel ensemble-based quantum information metric tensor, known as E-QFIM \(\mqfi\). This metric relies on a covariance structure and is designed for efficient estimation without bias, the need for quantum state replication, or full-state tomography. 
Additionally, we proposed the 2-QNSCD optimization method, which focuses solely on the underlying geometry of the space of quantum states corresponding to a random pair of parameter coordinates at each iteration. By using just this pair of coordinates, we can construct an unbiased estimator of the E-QFIM. This approach significantly reduces the computational cost associated with evaluating \(O(\nparams^2)\) terms at every iteration.

\section{Convergence Analysis of QNGD}
In this section, we present the convergence analysis of the vanilla QNGD, assuming complete access to a quantum information metric tensor and gradients. Traditionally, exponential convergence of GD has been proven for a certain class of functions, such as strongly convex functions and the functions that satisfy PL inequality (see Definition \ref{def:PLinequality} below) \cite{gower2018convergence,karimi2016linear,polyak1963gradient}. 

\begin{definition}
    [Polyak-Lojasiewicz (PL) Inequality]\label{def:PLinequality}
    A function $\Loss:\RR^c \rightarrow \RR$ is said to satisfy the PL inequality if for all $\param \in \RR^c$, the following inequality holds:
\begin{equation}\label{eqn:PLinequality}
    \frac{1}{2}\norm{\gradL(\param)}^2 \geq \mu (\Loss(\param) -\Loss(\param^*)) \quad \text{for some $\mu >0$},
\end{equation}
where $\param^* \in \RR^c$ is a global minimizer for $\Loss$.
\end{definition}

However, despite the advantages of NGD over GD, there are limited analytical proofs of convergence that have been established for NGD, for example, for strongly convex loss functions in classical neural networks \cite{martens2020new,zhang2019fast}, whereas almost no formal proof exists within the quantum setup. Therefore, this motivates us to characterize a metric-dependent sufficient condition, \textit{quadratic geometric information} (QGI) inequality, that ensures an exponentially faster rate of convergence.

\begin{definition}[Quadratic Geometric Information (QGI) Inequality]\label{eqn:QGI_inequality}
    For a given metric tensor $\fubini$, a function $\Loss:\RR^c \rightarrow \RR$ is said to satisfy the QGI inequality if the following inequality holds for some $\mu >0$ and for all $\param \in \RR^c$:
\begin{equation}
    \frac{1}{2}\gradL^{\ttT}(\param)\mathsf{F}(\param)^{-1} \gradL(\param) \geq \mu (\Loss(\param) -\Loss(\param^*)),
\end{equation}
where $\param^* \in \RR^c$ is a global minimizer for $\Loss(\param)$.
\end{definition}
To further understand the significance of the QGI inequality, consider an example from \cite{amari1998whyng}. \begin{example} Consider the following non-convex loss function using the polar coordinates: 
\begin{equation*}
    \Loss(r,\theta) = \frac{1}{2}[(r\cos(\theta)-1)^2 + r^2\sin^2(\theta)],
\end{equation*}
    where $r\geq 0 \eqand \theta < |\pi|$. The stationary points of $\Loss$ are $(1,0), (0,\pi/2), \eqand (0,-\pi/2)$ with $(1,0)$ being the global minimum. The Riemannian metric tensor for polar coordinates can be written as:
   \begin{equation*}
       \mathsf{F}(r,\theta) = \begin{pmatrix}
        1 &0\\ 0 &r^2
    \end{pmatrix}.
   \end{equation*}
Clearly, the above loss function does not satisfy the PL inequality because for saddle points $(0,\pi/2)$ or $ (0,-\pi/2)$, \eqref{eqn:PLinequality} does not hold for any $\mu > 0$. Conversely, $\Loss$ satisfies the QGI inequality for all 
$(r,\theta) \in [0,\infty) \!\times\! (-\pi,-\pi)$ 
because 
\begin{equation*}
    \lim_{(r,\theta)\rightarrow(0,\pi/2)} \gradL(r,\theta)^\ttT\mathsf{F}^{-1}(r)\gradL(r,\theta) = \lim_{(r,\theta)\rightarrow(0,-\pi/2)} \gradL(r,\theta)^\ttT\mathsf{F}^{-1}(r)\gradL(r,\theta) = 1.
\end{equation*}
This implies if we choose $\mu \in (0,1],$ the QGI inequality holds. 
\end{example} 
\noindent Furthermore, the VQE-inspired example in Fig.~\ref{fig:QNGDexample}, also satisfies only QGI inequality (for more details, see Supplementary Note 3). This implies that QNGD can provide an exponential rate of convergence when GD fails to do so. Below, we provide the exponential convergence theorem for QNGD under the following assumption:

\noindent\textbf{Assumption} ($\mathsf{L}$-smooth with respect to a Quadratic Norm \cite{nesterov2018lectures})
    For a given metric tensor $\fubini$, the function $\Loss(\param)$ is $\mathsf{L}$-smooth  with respect a quadratic norm $\|\param\|_{\text{$\mathsf{F}$}} \deq (\param^\ttT \mathsf{F}  \param )^{1/2}$, i.e., for all $\param, \param_1,\param_2 \in \RR^c$, the following inequalities hold:
 \begin{align}
        &\text{(a)} \ \Loss(\param_2) \leq \Loss(\param_1) + \gradL(\param_1)^\ttT(\param_2-\param_1) + \frac{\mathsf{L}}{2}\norm{\param_2 - \param_1}_{\text{$\mathsf{F}(\param_1)$}},\label{eqn:LsmoothQuadratic1}\\
        &\text{(b)} \ \|\gradL(\param)\|_{\mathsf{F}(\param)^{-1}}^2 \leq {2}\mathsf{L}(\Loss(\param)-\Loss(\param^*)),\label{eqn:LsmoothQuadratic2}
    \end{align}
for some $\mathsf{L} >0$, where $\param^*$ is the global minimum of $\Loss(\param)$.
\begin{theorem}[Convergence of QNGD]\label{thm:QNGDconv}
    Consider a $\mathsf{L}$-smooth (with respect to $\|\cdot\|_\mathsf{F}$) loss function $\Loss(\param)$ that satisfies the QGI inequality $\eqref{eqn:QGI_inequality}$, for some $\mu >0$. Let $\paramstar \in \RR^c$ be the global minimum of $\Loss(\param)$. Then, QNGD with a fixed learning rate $\eta = 1/\mathsf{L}$ and the update rule:
     \begin{equation*}
         \param^{(\ttt+1)} = \param^{(\ttt)} - \eta \mathsf{F}(\param^{(\ttt)})^{-1}\gradL(\param^{(\ttt)}),
     \end{equation*} 
     achieves a global exponential convergence rate, given by 
     \begin{equation}\label{eq:SQNGDconvergence_rate}
     \Loss(\param^{(\ttt)}) - \Loss(\paramstar) \leq \Big(1-\frac{\mu}{\mathsf{L}}\Big)^\ttt (\Loss(\param^{(0)}) - \Loss(\paramstar)).\end{equation} 
\end{theorem}
\begin{proof}
    We begin with \eqref{eqn:LsmoothQuadratic1} and apply the update rule.
    Consider the following inequalities:
    \begin{align*}
        \Loss(\param^{(\ttt+1)})  &\overset{}{\leq}   \Loss(\param^{(\ttt)}) 
        -\eta \nabla  \Loss(\param^{(\ttt)}) ^{\ttT} 
 \mathsf{F}(\param^{(\ttt)})^{-1}\gradL(\param^{(\ttt)}) +\eta^2 \frac{\mathsf{L}}{2} \| \mathsf{F}(\param^{(\ttt)})^{-1}\gradL(\param^{(\ttt)})\|_{\mathsf{F}(\param^{(\ttt)})}^2 \\
 &\overset{a}{=} \Loss(\param^{(\ttt)})  -\frac{1}{2\mathsf{L}} \nabla  \Loss(\param^{(\ttt)}) ^{\ttT} 
 \mathsf{F}(\param^{(\ttt)})^{-1}\gradL(\param^{(\ttt)}) \\
&\overset{b}{\leq}  \Loss(\param^{(\ttt)}) -\frac{\mu}{\mathsf{L}}(\Loss(\param^{(\ttt)}) - \Loss(\param^*)),
\end{align*}
where $(a)$ follows by putting $\eta = 1/\mathsf{L}$ and from the definition of the quadratic norm, and $(b)$ follows from QGI inequality \eqref{eqn:QGI_inequality}.
Thus, after re-arranging and subtracting $\Loss(\param^*)$ from both sides, we get,
\begin{align*}
\Loss(\param^{(\ttt+1)}) -  \Loss(\param^*) &\leq \Big(1-\frac{\mu}{\mathsf{L}}\Big) (\Loss(\param^{(\ttt)}) -  \Loss(\param^*)). \end{align*}
Applying this inequality recursively, we get,
\[ \Loss(\param^{(\ttt)}) - \Loss(\param^*) \leq \Big(1-\frac{\mu}{\mathsf{L}}\Big)^\ttt (\Loss(\param^{(0)}) - \Loss(\param^*)).\]
Note that from the QGI inequality and \eqref{eqn:LsmoothQuadratic2}, we conclude that $0 < {\mu}/{\mathsf{L}} < 1$. This completes the proof of Theorem \ref{thm:QNGDconv}.  
\end{proof}
The proof is notably simple and does not require $\Loss$ to be convex or strongly convex. In addition, this is a significant general result for achieving an exponential convergence rate using metric-dependent optimization methods for non-convex problems.
\section{Illustrating QNGD convergence.}
In this section, we illustrate an advantage of QNGD over GD by presenting a one-qubit VQE-inspired example that satisfies the QGI inequality (Def. \ref{eqn:QGI_inequality}), whereas it violates the PL inequality (Def. \ref{def:PLinequality}). 
We illustrate that QNGD has the potential to escape saddle points, while GD tends to get trapped at those points.
\begin{figure}
    \centering
    \includegraphics[height=4.6in,width =\textwidth]{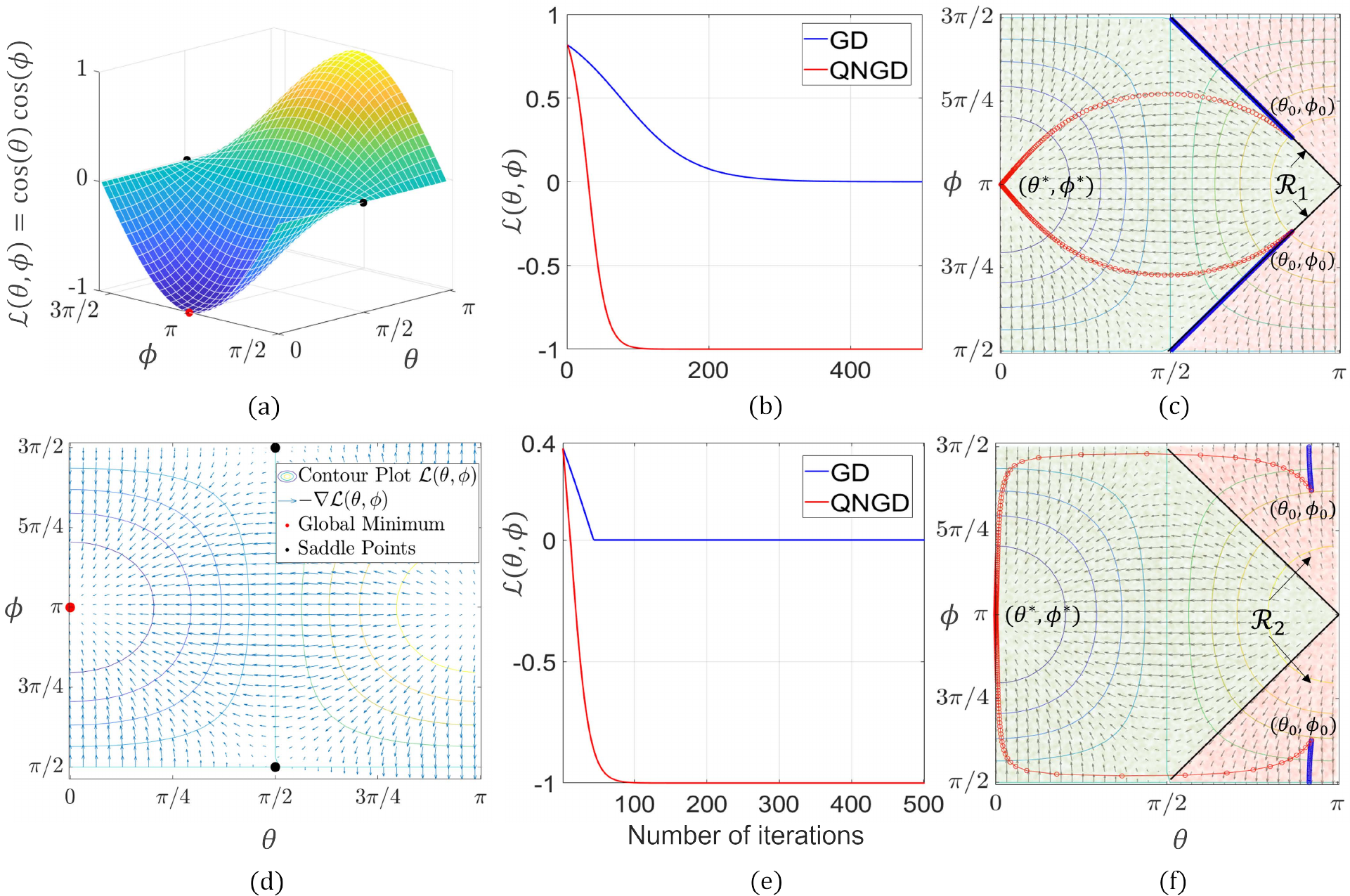}
    \caption{Comparison of GD and QNGD in finding the ground state of the Pauli-X matrix: QNGD escapes saddle points and reaches the global minimum more efficiently than GD. 
    }
    \label{fig:QNGDexample}
\end{figure}

Consider a single-qubit parameterized unitary
\begin{equation*}
    U(\theta,\phi) \deq \begin{pmatrix}
    \cos(\theta/2) & -e^{i\phi}\sin(\theta/2)\\
    \sin(\theta/2) & e^{i\phi}\cos(\theta/2)
\end{pmatrix},
\end{equation*}
where the parameters' domain is defined as $\calR \deq \set{(\theta,\phi):0\leq \theta < \pi \eqand \pi/2 \leq \phi \leq 3\pi/2}$.
The goal is to optimize the parameters of the state $\ket{\psi(\theta,\phi)} \deq U(\theta,\phi)\ket{\psi_0}$, using QNGD and GD optimization methods, to reach the ground state of the Hamiltonian $H = \sigma_X$, starting from the initial state $\ket{\psi_0} = \left(\frac{|0\>+|1\>}{\sqrt{2}}\right)$. The loss function (or energy function) is given as $\Loss(\theta,\phi) \deq \expval{\sigma_X}{\psi(\theta,\phi)} = \cos(\theta)\cos(\phi)$, and is shown in Fig.~1(a). The corresponding contour plot with vector field of steepest descent direction in the Euclidean parameter space, i.e., $-\gradL(\theta, \phi)$ is shown in Fig.~\ref{fig:QNGDexample}(d). The stationary points are characterized by $\gradL(\theta,\phi) = 0 \text{ for } (\theta,\phi) \in \calR$. This gives, a global minimum $(\theta^*,\phi^*) = (0,\pi)$ and saddle points $(\pi/2,\pi/2) \eqand (\pi/2,3\pi/2)$. In Fig~\ref{fig:QNGDexample}(a) and (d), the global minimum is marked as a red dot, and the saddle points are marked in black. Using Def. \ref{def:fubinistudymetric}, the Fubini-Study metric is calculated as: 
\begin{equation*}
    \fubini = \frac{1}{4}\begin{pmatrix}
    {\cos^2(\phi)} &0\\
    0 & 1
\end{pmatrix}.
\end{equation*} 
We compare the convergence of QNGD and GD against the number of iterations and for a fixed learning rate $\eta = 0.01$. Convergence for GD (blue) and QNGD (red) against the number of iterations for initial points in regions $\calR_1 \eqand \calR_2$ 
    are shown in Fig.\ref{fig:QNGDexample}(b) and (e), respectively. The parameters trajectory taken by GD (blue) and QNGD (red) starting from an initial point in regions $\calR_1$ and $\calR_2$ are shown in Fig.\ref{fig:QNGDexample}(c) and (f), respectively.
    In particular, we note the following:

\noindent $\bullet$ For all initial points in the region \( \mathcal{R}_1 \deq \{(\theta_0, \phi_0) \in \mathcal{R} : \theta_0 = \phi_0 \text{ or } \theta_0 + \phi_0 = 2\pi\} \), GD tends to get trapped at saddle points. This occurs because the steepest descent direction in the Euclidean (or $\ell_2$) geometry, given by \(-\gradL\), guides towards these saddle points for all \((\theta, \phi) \in \calR_1\). In contrast, QNGD follows a different trajectory, allowing it to avoid these saddle points. Considering the geometry of the space of quantum states, the direction of steepest descent given by \(-\fubini^{-1}\gradL\) tends to bend away from saddle points and converges to the global minimum, as illustrated in Fig.~1(b) and (c).

\noindent $\bullet$ Let $\calB_1 \deq \set{(\theta,3\pi/2):\theta\in [0,\pi)}$ and $\calB_2 \deq \set{(\theta,\pi/2):\theta\in [0,\pi)}$ denote the top and bottom boundaries of the region $\calR$, respectively. For all initial points in the region $\calR_2 \deq \set{(\theta_0, \phi_0) \in \mathcal{R}: \theta_0 > \phi_0 \text{ or } \theta_0 + \phi_0 > 2\pi}$, GD fails to converge to a global minimum, as shown in Fig.~1(e) and (f). This is because $-\gradL$ directs towards boundaries $\calB_1\eqand\calB_2$, and thus getting trapped indefinitely. Conversely, QNGD, by following the steepest descent direction in the geometry of the space of quantum states, successfully converges to the global minimum for all $(\theta_0, \phi_0) \in \mathcal{R}_2$. This is because the top and bottom boundaries of the region $\mathcal{R}$ correspond to singular points of the Fubini-Study metric, i.e., $\det(\fubini) = 0$ for all $(\theta, \phi) \in (\calB_1 \cup \calB_2)$.
    
    The quantum states corresponding to the singular region $\calB_1$, given by $e^{i\theta/2} \left(\frac{|0\rangle - i |1\rangle}{\sqrt{2}}\right)$, are indistinguishable, meaning the fidelity between any two quantum states in $\calB_1$ is one. As a result, the loss function remains constant for all parameterized quantum states within $\calB_1$, irrespective of the Hamiltonian. Similarly, for the singular region $\calB_2$, the quantum states, given by $e^{-i\theta/2} \left(\frac{|0\rangle + i |1\rangle}{\sqrt{2}}\right)$, are also indistinguishable, and loss function remains constant over $\calB_2$ for any Hamiltonian. This implies that for an arbitrary Hamiltonian that leads to a non-constant loss function over $\mathcal{R}$, there is a possibility of finding a global (or local) minimum inside $\mathcal{R}$. Therefore, the QNGD prevents the parameter trajectory from approaching the boundaries of $\mathcal{R}$. In other words, when the parameter trajectory gets near these singular regions, the volume of the metric contracts along the axis corresponding to the parameter $\theta$, and the QNGD stretches the descent direction along the same axis and guides it toward the global minimum.

 \noindent \textbf{Convergence using GD}: For all $(\theta, \phi) \in \mathcal{R}$, the function $\Loss(\theta, \phi)$ does not satisfy the PL inequality for any $\mu > 0$. This is because the PL inequality implies that every stationary point is a global minimum. However, when the loss function contains saddle points, this condition fails. For these points, we have $\|\gradL(\theta,\phi)\|^2 = 0$, but $(\Loss(\theta,\phi) -\Loss(\theta^*,\phi^*)) \neq 0$. This means that GD fails to guarantee convergence to the global minimum, as illustrated in the example above.

\noindent \textbf{Convergence using QNGD}: The two saddle points are correctly characterized by the singular points of $\fubini$, i.e., $\det(\fubini(\pi/2,\pi/2)) = \det(\fubini(\pi/2,3\pi/2)) = 0$. Thus, for saddle points, we have \begin{equation*}
    \lim_{(\theta,\phi)\rightarrow (\frac{\pi}{2},\frac{\pi}{2})} \gradL(\theta,\phi)\fubini^{-1}(\phi)\gradL(\theta,\phi) = \lim_{(\theta,\phi)\rightarrow (\frac{\pi}{2},\frac{3\pi}{2})} \gradL(\theta,\phi)\fubini^{-1}(\phi)\gradL(\theta,\phi) = 1.
\end{equation*} This implies if we choose $\mu \in (0,1/2]$, the QGI inequality holds for these points. Therefore, for all $(\theta, \phi) \in \mathcal{R}$, $\Loss(\theta, \phi)$ does satisfy the QGI inequality for a sufficiently small $\mu \in (0,1/2]$. This signifies that QNGD can provide an exponential convergence rate to a global minimum even for multi-modal functions with several local-saddle points, particularly encountered in VQAs, whereas GD fails to provide a guaranteed exponential rate of convergence.

\section{Properties of Ensemble Distance}
\noindent \textit{$\bullet\ d_{\text{E}}(\calE(\param),\calE(\param')) = 0$ if and only if the corresponding quantum state within $\calE(\param)\eqand \calE(\param')$ are same up to a constant global phase.} Consider the following inequalities: 
\begin{equation*}
    \big|\sum_{x\in \calX} \Qx(x) \<\phi_x(\param)|\phi_x(\param^{'})\>\big| \overset{a}{\leq} \sum_{x\in \calX} \Qx(x) \big|\<\phi_x(\param)|\phi_x(\param^{'})\>\big| \overset{b}{\leq}\sum_x \Qx(x) = 1,
\end{equation*} where $(a) \eqand (b)$ follows from Jensen's inequality and Cauchy–Schwarz inequality, respectively. Note that the first equality holds if and only if $\<\phi_x(\param)|\phi_x(\param^{'})\> = r_x e^{i\delta}$ for all $x\in \calX$ and second equality holds if and only if $|\phi_x(\param^{'})\>= e^{i\delta_x}  |\phi_x(\param)\>$ for all $x \in \mathcal{X}$, where $r_x \geq 0, \eqand \delta, \delta_x \in [0,2\pi)$. This implies $d_{\text{E}} = 0$ iff $|\phi_x(\param')\> = e^{i\delta}|\phi_x(\param')\>$ for all $x\in \calX$ and $\delta \in [0,2\pi).$

\noindent \textit{$\bullet\ d_{\text{E}}$ is monotonic.} Consider the following ensembles of pure states $\calE_\phi \deq \set{(\Qx(x),\ket{\phi_x})} \eqand \calE_\psi \deq \set{(\Qx(x),\ket{\psi_x})}$. The action of the quantum channel $\calN$ on the quantum ensemble of pure states is described using the convex linear postulate of a quantum channel as follows: $\calN(\calE_\phi) = \sum_{x\in\calX} \Qx(x) \ \calN(\Phi_x).$ This implies that a quantum channel acting on an ensemble of pure states produces a density matrix of the form $\rho^\phi \deq \sum_x\Qx(x)\rho^\phi_x$, where $\rho^\phi_x \deq \calN(\Phi_x).$ Similarly, $\rho^\psi \deq \sum_x\Qx(x)\rho^\psi_x$. Therefore, to show that $d_{\text{E}}$ is monotonic, we need to establish the following: 
\begin{equation*}
    d_{\text{B}}(\rho^\phi,\rho^\psi) \leq d_{\text{E}}(\calE_\phi,\calE_\psi),
\end{equation*}
where $d_{\text{B}}(\rho^\phi,\rho^\psi)$ is the Bures distance \eqref{eqn:buresdistance} between $\rho^\phi$ and $\rho^\psi$. 

Consider the following inequalities:
\begin{align}
d_{\text{B}}^2(\rho^\phi,\rho^\psi) = 2-2(f_{\rho}(&\rho^\phi,\rho^\psi))^{1/2}
\overset{a}{\leq} 2-2\Big(\sum_x\Qx(x)f_\rho{(\rho^\phi_x,\rho^\psi_x)}\Big)^{\tfrac{1}{2}}\overset{b}{\leq} 
2-2\Big(\sum_x\Qx(x)|\<\phi_x|\psi_x\>|^2\Big)^{\tfrac{1}{2}}\nonumber\\
&\overset{c}{\leq} 2-2\Big|\sum_x\Qx(x)\<\phi_x|\psi_x\>\Big| = d_{\text{E}}^2(\calE_\phi,\calE_\psi),\nonumber
\end{align}
where $(a)$ follows from the joint concavity of Uhlmann fidelity \cite[Ch.9]{wilde2013quantum} and the fact that square root is a monotonically increasing function, $(b)$ follows from the fact that Uhlmann fidelity is monotone with respect to the channel $\calN$, and $(c)$ follows from the fact that $|z|^2$ is a convex function, where $z$ is a complex number, and by using Jensen's inequality. This completes the proof that the Bures distance between the channel outputs is less than or equal to the ensemble distance, thereby showing $d_{\text{E}}$ is monotonic.

\section{Proof of Theorem 1.}
We start the derivation of the E-QFIM by writing the squared ensemble distance between two infinitesimally close ensembles $\calE(\param) \eqand \calE(\param+\dparam),$ which is given as 
\begin{equation}
    d_{\text{E}}^2(\calE(\param),\calE(\param+\dparam)) = 2-2\sqrt{f_\calE(\param,\param+\dparam)} 
    .\label{eqn:derEQFIM_0}
\end{equation}
Now, assume $\param_i$ denotes parameter corresponds to ($\a^\textth$ layer, $\p^\textth$ qubit), and $\param_j$ denotes parameter corresponds to ($\b^\textth$ layer, $\q^\textth$ qubit). Then, the Taylor series of $\ket{\phi(\param+\dparam)}$ (up to the first order) is written as:
\begin{equation*}
    \ket{\phi(\param+\dparam)} = \ket{\phi(\param)} + \sum_{i=1}^c \ket{\partial_i \phi(\param)} \dparam_i 
.
\end{equation*}
Therefore, we can write the fidelity between $\calE(\param) \eqand \calE(\param+\dparam)$ as:
\begin{align}
    f_\calE(\param,\param+\dparam) &= \Big|\sum_{x\in\calX} \Qx(x) \Big(\<\phi_x(\param)|\phi_x(\param)\> + \sum_{i=1}^c \dparam_i\<\phi_x(\param)|\partial_i\phi_x(\param)\>  \Big)\Big|^2\nonumber\\
    &= 1 + \sum_{x\in\calX} \Qx(x) \sum_{i=1}^c \dparam_i\Big(\<\phi_x(\param)|\partial_i\phi_x(\param)\> + \<\partial_i \phi_x(\param)|\phi_x(\param)\>\Big)\nonumber\\
    &\hspace{20pt}+
    \sum_{x_1  x_2} \Qx(x_1)\Qx(x_2) \sum_{i j} \dparam_i\dparam_j \<\partial_i \phi_{x_1}(\param)|\phi_{x_1}(\param)\>\<\phi_{x_2}(\param)|\partial_j\phi_{x_2}(\param)\>.\label{eqn:derEQFIM_1}
\end{align}
Using 
$f_\calE(\param+\dparam,\param+\dparam) = 1$, we obtain
\begin{align}
    \Big|\sum_{x\in\calX}\Qx(x) \<\phi_x(\param\!+\!\dparam)&|\phi_x(\param\!+\!\dparam)\> \Big| 
    =  \Big|\sum_{x\in\calX}\Qx(x) \Big[1+ \sum_{i} \dparam_i\big(\<\phi_x(\param)|\partial_i\phi_x(\param)\>\! +\! \<\partial_i \phi_x(\param)|\phi_x(\param)\>\big) \nonumber\\
    &+\sum_{i j} \dparam_i\dparam_j \<\partial_i \phi_{x}(\param)|\partial_j\phi_{x}(\param)\>\Big]\Big| = 1. \label{eqn:derEQFIM_2}
\end{align}
The above equality implies the following relation:
\begin{equation}
    \sum_{i} \dparam_i\big(\<\phi_x(\param)|\partial_i\phi_x(\param)\>\! +\! \<\partial_i \phi_x(\param)|\phi_x(\param)\>\big) = -\sum_{i j} \dparam_i\dparam_j \<\partial_i \phi_{x}(\param)|\partial_j\phi_{x}(\param)\>.\label{eqn:derEQFIM_3}
\end{equation}
\vspace{-10pt}
Finally, considering \eqref{eqn:derEQFIM_0}, \eqref{eqn:derEQFIM_1}, and \eqref{eqn:derEQFIM_3} collectively, we see that 
\begin{align}
    \mqfi_{(i,j)}(\param)&= \frac{1}{2}\frac{\partial^2 }{\partial \dparam_i \partial \dparam_j} d_{\text{E}}^2(\param,\param + \dparam) \bigg\vert_{\dparam = 0}\nonumber\\
    &=\frac{1}{2}\frac{\partial^2 }{\partial \dparam_i \partial \dparam_j} \bigg[2 - 2\Big(1-\sum_{x\in\calX}\Qx(x)\sum_{i j} \dparam_i\dparam_j \<\partial_i \phi_{x}(\param)|\partial_j\phi_{x}(\param)\> \nonumber\\
    & \hspace{20pt}+\sum_{x_1, x_2} \Qx(x_1)\Qx(x_2) \sum_{i j} \dparam_i\dparam_j \<\partial_i \phi_{x_1}(\param)|\phi_{x_1}(\param)\>\<\phi_{x_2}(\param)|\partial_j\phi_{x_2}(\param)\> \Big)^{1/2}\bigg]\bigg\vert_{\dparam = 0}  \nonumber \\
    &\overset{a}{=} \text{Re} \Big\{\sum_{x\in\calX}\Qx(x)\<\partial_i \phi_{x}(\param)|\partial_j\phi_{x}(\param)\> - \sum_{x_1,  x_2} \Qx(x_1)\Qx(x_2) \<\partial_i \phi_{x_1}(\param)|\phi_{x_1}(\param)\>\<\phi_{x_2}(\param)|\partial_j\phi_{x_2}(\param)\>\Big\}\nonumber\\
    &\overset{b}{=}\Re{\Tr(\frac{\partial U^{\dagger}(\param)}{\partial \param_{i}} \frac{\partial U(\param)}{\partial \param_{j}}\rho) - \Tr(\frac{\partial U^{\dagger}(\param)}{\partial \param_{i}} U(\param)\rho)\Tr(U^{\dagger}(\param) \frac{\partial U(\param)}{\partial \param_{j}}\rho)}\nonumber\\
    &=\Re{\Tr(\Upsilon_{i}(\param)\Upsilon_{j}(\param)\rho) - \Tr(\Upsilon_{i}(\param)\rho)\Tr(\Upsilon_{j}(\param)\rho)} = \Cov(\Upsilon_{i}(\param),\Upsilon_{j}(\param))_{\rho},
    \end{align}
    where in $(a)$, the real part appears from the fact that $\dparam_i\dparam_j$ occurs twice in the summation but with the conjugated terms, and $(b)$ follows from the cyclicity of trace. This completes the proof of Theorem 1.
    
\section{Proof of Lemma 1.}
\noindent Consider the parameterized unitary of the $\a^{\textth}$ hidden layer, $U_\a(\param_\a) = \bigotimes_{p=1}^d R_{\sigmaap}(\param_{(\a,\p)})$. Recalling the notation from Introduction section, the derivative of  $W_\a(\param_\a)$ with respect to $\p^\textth$ qubit parameter can be expressed as:
\begin{equation*}
        \partial_{\ap} W_{\a}(\param_\a) \deq \frac{\partial W_\a(\param_\a)}{\partial {\param_{(\a,\p)}}} = -\frac{i}{2} W_\a(\param_\a) (\II^{\tensor [1:\p)} \tensor \sigmaap \tensor \II^{\tensor (\p:d]}).
\end{equation*}
Further, the derivative of $U(\param)$ can be written as:
\begin{align}
        \partial_{\ap} U(\param) \deq \frac{\partial U(\param)}{\partial {\param_{(\a,\p)}}} 
        &= -\frac{i}{2}  W_{[\bfL:\a]} (\II^{\tensor [1:\p)} \tensor \sigmaap \tensor \II^{\tensor (\p:d]}) W_{(\a:1]}, \label{eqn:QNNgrad}
\end{align}
Using \eqref{eqn:QNNgrad}, we compute the elements of the E-QFIM corresponding to the ($\a^{\text{th}}$ layer, $\p^{\text{th}}$ qubit) parameter and ($\b^{\text{th}}$ layer, $\q^{\text{th}}$ qubit) parameter. Consider the second term of $\mqfi$ (see (11)). 
\begin{align}
        &\Tr\set{(\partial_{\ap} U^{\dagger}(\param)) U(\param)\rho}\Tr\set{U^{\dagger}(\param) (\partial_{\bq} U(\param))\rho} \nonumber \\
        &=0.25 \! \cdot \!\Tr\set{W_{[1:\a)}^{\dagger}(\II^{\tensor [1:\p)} \tensor \sigmaap \tensor \II^{\tensor (\p:d]})W_{[\a:\bfL]}^{\dagger} W_{[\bfL:1]} \rho}\Tr\set{W_{[1:\bfL]}^{\dagger} 
    W_{[\bfL:\b]}(\II^{\tensor [1:\q)} \tensor \sigmabq \tensor \II^{\tensor (\q:d]}) W_{(\b:1]}\rho}\nonumber\\
    &=0.25 \! \cdot\! \Tr\set{W_{[1:\a)}^{\dagger}(\II^{\tensor [1:\p)} \tensor \sigmaap \tensor \II^{\tensor (\p:d]}) W_{(\a:1]} \rho}\Tr\set{W_{[1:\b)}^{\dagger} 
   (\II^{\tensor [1:\q)} \tensor \sigmabq \tensor \II^{\tensor (\q:d]}) W_{(\b:1]}\rho}\nonumber\\
   &= \Tr\set{\Upsilon_{\ap}\rho} \Tr\set{\Upsilon_{\bq} \rho}.
\end{align} 
To compute the first term of (11), without loss of generality, assume $\a \leq \b$. Thus, we get
\begin{align}
        &\Re{\Tr\tuple{(\partial_{\ap} U^{\dagger}(\param))(\partial_{\bq} U(\param))\rho}} \nonumber\\
        &= 0.25 \cdot \Re{\Tr\tuple{W_{[1:\a)}^{\dagger} (\II^{\tensor [1:\p)} \tensor \sigmaap \tensor \II^{\tensor (\p:d]}) W_{[\a:\b)}^{\dagger} (\II^{\tensor [1:\q)} \tensor \sigmabq \tensor \II^{\tensor (\q:d]}) W_{(\b:1]} \rho}}\nonumber\\
        &= 0.25 \cdot \Re{\Tr\tuple{W_{[1:\a)}^{\dagger} (\II^{\tensor [1:\p)} \tensor \sigmaap \tensor \II^{\tensor (\p:d]}) W_{(\a:1]} W_{[1:\b)}^{\dagger} (\II^{\tensor [1:\q)} \tensor \sigmabq \tensor \II^{\tensor (\q:d]}) W_{(\b:1]} \rho}}\nonumber\\
        &=\Re{\Tr\tuple{\Upsilon_{\ap}\Upsilon_{\bq}\rho}}= \frac{1}{2}\Tr\tuple{\set{\Upsilon_{\ap},\Upsilon_{\bq}} \rho},
\end{align}
where the last equality follows from the cyclicity of the trace.
Combining these expressions, the E-QFIM takes the form: $\calF_{(\ap,\bq)}(\param) = \Cov(\Upsilon_{\ap},\Upsilon_{\bq})_{\rho},$ where $\Upsilon_{\ap} = W_{[1:\a)}^{\dagger}(\II^{\tensor [1:\p)} \tensor (\sigmaap/2) \tensor \II^{\tensor (\p:d]}) W_{(\a:1]}$ and $\Upsilon_{\bq} = W_{[1:\b)}^{\dagger}(\II^{\tensor [1:\q)} \tensor (\sigmabq/2) \tensor \II^{\tensor (\q:d]}) W_{(\b:1]}.$ 

\section{Proof of Lemma 2.}
\noindent 
Using the cyclicity of trace, we rewrite the expectation of $\set{\bfA,\bfB}$ with respect to the state $\rho$ as follows:
\begin{align}
\vspace{-10pt}
 \Tr(\set{\bfA,\bfB} \rho) &= \Tr(\bfB \rho \bfA + \bfB \bfA \rho) = \Tr(\set{\bfA,\rho}\bfB). \label{eqn:tracecyclicity_anticomm}
\end{align}
Suppose $A$ is a random variable denoting the measurement outcome of $\calM = \set{\bfA_+,\bfA_-}$ on the state $\rho$. The probability of obtaining outcome $1$, i.e., $\prob{A =1} = \Tr(\bfA_+ \rho)$. Similarly, $\prob{A = -1} = \Tr(\bfA_- \rho)$. Then, we decompose the anit-commutator $\set{\bfA,\rho}$ in terms of $\bfA_+  \eqand \bfA_-$ as follows:
\begin{align*}
    \set{\bfA,\rho} = \bfA \rho I + I \rho \bfA &\overset{a}{=}(\bfA_+ - \bfA_-) \rho (\bfA_+ + \bfA_-) + (\bfA_+ + \bfA_-)\rho (\bfA_+ - \bfA_-) \\
    &=2[\bfA_+ \rho  \bfA_+  - \bfA_-\rho \bfA_-] \\
    &= 2[\prob{A=1} \rho(\bfA_+) - \prob{A=-1} \rho(\bfA_-)],
\end{align*}
where $(a)$ follows from the facts that the Hermitian matrix has orthonormal eigenvectors for distinct eigenvalues and $\bfA^2 = \II$. Here, $\rho(\bfA_+) \eqand \rho(\bfA_-)$ denote the post-measurement state when we measure outcomes $+1 \eqand -1$, respectively. Subsequently, let $B$ be the random variable denoting the measurement outcome of the post-measurement state along the eigenvectors of $\bfB$. Then,
\begin{align}
    \frac{1}{2}\Tr(\set{\bfA,\rho}\bfB) &= \prob{A=1} \Tr(\bfB\rho(\bfA_+)) - \prob{A=-1} \Tr(\bfB\rho(\bfA_-)) \nonumber\\
    &= \prob{A=1}\EE[AB|A=1] + \prob{A=-1}\EE[AB|A=-1]=\EE[AB]. \label{eqn:trace_EAB}
\end{align}
Hence, from \eqref{eqn:tracecyclicity_anticomm}  and \eqref{eqn:trace_EAB}, we get $\Tr(\set{\bfA,\bfB}\rho) = 2\cdot \EE[AB]$. This completes the proof of Lemma 2.

\section{Proof of Lemma 3.}
Let $|\phi_{x}^{a}\> \deq W_{(\a:1]}\ket{\phi_{x}}$ denote the input state to the $\a^{\text{th}}$ layer of the PQC. We denote $\Sigmaap \deq (\II^{\tensor [1:\p)} \tensor \sigmaap \tensor \II^{\tensor (\p:d]})$ and $\Sigmabq \deq (\II^{\tensor [1:\q)} \tensor \sigmabq \tensor \II^{\tensor (\q:d]})$. We first compute the expectation of diagonal entries. Observe that $\EE[v_i|\ketphi{x_i}] = \Tr(\Sigmabq \ \Phi^\b_{x_i})$ for $i = {1,2}$. For $z_{22}$, 
\begin{align}
    \EE[z_{22}] = \EE[0.25\cdot(1-v_1v_2)] &\overset{a}{=} 0.25\cdot(1-\EE_{X_1}[\EE[v_1|\ketphi{x_1}]]\EE_{X_2}[\EE[v_2|\ketphi{x_2}]]) \nonumber\\
    &= 0.25\cdot(1-\EE_{X_1}[\Tr\parens{\Sigmabq \ \Phi^\b_{x_1}}]\EE_{X_2}[\Tr\parens{\Sigmabq\ \Phi^\b_{x_2}}])\nonumber \\
    &= 0.25-0.25[\Tr\parens{\Sigmabq \ W_{(\b:1]}\rho W^{\dagger}_{[1:\b)}}]^2 = 0.25-[\Tr(\Upsilon_{\bq}\rho)]^2,\nonumber
\end{align}
where $(a)$ follows because $\ketphi{x_1} \eqand \ketphi{x_2}$ are i.i.d. samples. This gives the desired result for $z_{22}$. 
Next, we note that the structure of the pair $(v_1,v_2)$ is exactly the same as that of $(u_1,u_2)$ with the correspondence $\ap \leftrightarrow \bq$. 
Hence, for $z_{11} = 0.25 \cdot (1-u_1u_2)$, we get $\EE[z_{11}] = 0.25-\Tr(\Upsilon_{\ap}\rho)^2$. 

Finally, we compute the expectation of off-diagonal entries. Recall the definition of Hermitian operators $\Upsilon_{\ap} \eqand \Upsilon_{\bq}$, then consider the following set of equalities:
\begin{align}
  \set{\Upsilon_{\ap},\Upsilon_{\bq}} &= \Upsilon_{\ap} \Upsilon_{\bq} +\Upsilon_{\bq}\Upsilon_{\ap}\nonumber\\
  &=  0.25\cdot(W_{[1:\a)}^{\dagger}\Sigmaap \ W_{(\a:1]} W_{[1:\b)}^{\dagger}\Sigmabq \ W_{(\b:1]} + W_{[1:\b)}^{\dagger}\Sigmabq \ W_{(\b:1]} W_{[1:\a)}^{\dagger}\Sigmaap \ W_{(\a:1]}) \nonumber \\
  &=  0.25\cdot(W_{[1:\a)}^{\dagger}\Sigmaap \ W_{[\a:\b)}^{\dagger}\Sigmabq \ W_{(\b:1]} + W_{[1:\b)}^{\dagger}\Sigmabq \ W_{(\b:\a]} \Sigmaap \ W_{(\a:1]}). \nonumber
\end{align}
Using the above expansion of $\set{\Upsilon_{\ap},\Upsilon_{\bq}}$ and cyclicity of trace, we get
\begin{align}
    \expval{\set{\Upsilon_{\ap},\Upsilon_{\bq}}}{\phi_x} &= \Tr(\set{\Upsilon_{\ap},\Upsilon_{\bq}} \Phi_x) \nonumber\\ 
    &= 0.25\cdot\Tr\parens{\Sigmaap \ W_{[\a:\b)}^{\dagger}\Sigmabq \ W_{(\b:1]} \Phi_x \ W_{[1:\a)}^{\dagger}} + \Tr\parens{\Sigmabq \ W_{(\b:\a]} \Sigmaap \ W_{(\a:1]} \Phi_x \ W_{[1:\b)}^{\dagger}}\nonumber\\
    &\overset{a}{=}0.25\cdot \Tr\parens{\Sigmaap \ W_{[\a:\b)}^{\dagger} \Sigmabq \ W_{(\b:\a]} \Phi^\a_x} + \Tr\parens{\Sigmabq \ W_{(\b:\a]} \ \Sigmaap \ \Phi^\a_x \ W_{[\a:\b)}^{\dagger}}\nonumber\\
    &= 0.25\cdot\Tr\parens{\Sigmabq \ W_{(\b:\a]} \Phi^\a_x \ \Sigmaap \ W_{[\a:\b)}^{\dagger}} + \Tr\parens{\Sigmabq \ W_{(\b:\a]} \ \Sigmaap \ \Phi^\a_x \ W_{[\a:\b)}^{\dagger}}\nonumber\\
    &= 0.25\cdot\Tr\parens{\Sigmabq \ W_{(\b:\a]} \set{\Sigmaap,\Phi^\a_x}  W_{[\a:\b)}^{\dagger}}, \label{eqn:QNNanticomm}
\end{align}
where $(a)$ follows from the definition $\Phi_x^a$. 
Using the following facts: (i) $\Sigmaap$ is Hermitian, (ii) $(\Sigmaap)^2= \II$, (iii) $\Tr\tuple{\set{\bfA,\bfB}\Phi} = \Tr\tuple{\bfB\set{\bfA,\Phi}}$, and Lemma 2,
we can write \begin{equation}\label{eqn:u1w1}
  \EE[u_iw_i|\ketphi{x_{i+2}}] = 0.5 \cdot \Tr\parens{W_{[\a:\b)}^{\dagger} \Sigmabq \ W_{(\b:\a]} \set{\Sigmaap,\Phi^\bfa_{x_{i+2}}}  },
\end{equation} for $i = 1,2$, where we used the Heisenberg picture of quantum measurements:
$\Tr(\Lambda U \rho U^{\dagger})=\Tr(U^{\dagger}\Lambda U \rho)$, where $\Lambda$ and $U$ represent a measurement and unitary operator, respectively \cite{sakurai2020modern}.
Now, we are equipped to compute the expectation of the off-diagonal entries. For $z_{12} = z_{21}$,
\begin{align*}
    \EE[z_{12}] &= 0.125\cdot \EE[u_1w_1 + u_2w_2] - 0.0625\cdot \EE\left[(u_1+u_2)(v_1+v_2)\right] \nonumber\\
    &\overset{a}{=} 0.125 \cdot\big(\EE_{X_3}[\EE[u_1w_1|\ketphi{x_3}]] + \EE_{X_4}[\EE[u_2w_2|\ketphi{x_4}]]\big)- \Tr(\Upsilon_{\ap}\rho)\Tr(\Upsilon_{\bq}\rho) \nonumber\\
    &\overset{b}{=} 0.0625 \cdot\big(\EE_{X_3}[\Tr\parens{\Sigmabq \ W_{(\b:\a]} \set{\Sigmaap,\Phi^\a_{x_3}}  W_{[\a:\b)}^{\dagger}}] + \EE_{X_4}[\Tr\parens{\Sigmabq \ W_{(\b:\a]} \set{\Sigmaap,\Phi^\a_{x_4}}  W_{[\a:\b)}^{\dagger}}]\big)\nonumber\\
    &\hspace{3.43in} - \Tr(\Upsilon_{\ap}\rho)\Tr(\Upsilon_{\bq}\rho)\nonumber\\
    &\overset{c}{=} 0.125 \cdot\Tr\parens{\Sigmabq \ W_{(\b:\a]} \set{\Sigmaap,\rho^\a}  W_{[\a:\b)}^{\dagger}}- \Tr(\Upsilon_{\ap}\rho)\Tr(\Upsilon_{\bq}\rho) \nonumber\\
    &\overset{d}{=} 0.5 \cdot  \Tr(\set{\Upsilon_{\ap},\Upsilon_{\bq}}\rho) - \Tr(\Upsilon_{\ap}\rho)\Tr(\Upsilon_{\bq}\rho), \nonumber
\end{align*}
where $(a)$ follows because $\ketphi{x_1} \ketphi{x_2},\ketphi{x_3}, \eqand \ketphi{x_4}$ are i.i.d. samples and the analysis of $z_{11},z_{22}$, $(b)$ follows from \eqref{eqn:u1w1}, $(c)$ follows by defining $\rho^\a \deq W_{(\a:1]}\rho W^{\dagger}_{[1:\a)}$, and $(d)$ follows from \eqref{eqn:QNNanticomm}.
This completes the proof of Lemma 3.

\section{Proof of Theorem 2.}
Without loss of generality, assume $\a\leq \b$.
Consider the following equalities:
\begin{align*}
    \EE[\bar{\bfZ}{(\ap,\bq)}] = \frac{\nparams(\nparams-1)}{2}\EE\left[\tilde{\bfZ}(\ap,\bq) - \frac{2\beta}{\nparams}\II\right] 
    \overset{a}{=}\frac{1}{2}\sum_{\ap \neq \bq} \EE\left[\tilde{\bfZ}{(\ap,\bq)}\right] - (\nparams-1)\beta\II \quad \quad
\end{align*}
\begin{align}
&\overset{b}{=} \frac{1}{2}\sum_{\ap \neq \bq} \left[
\begin{array}{cccccccccccc}
0_{\ddots} & & &  &&&&&\\
& 0 & & & &&&& \\
&  & \frac{1}{(\nparams-1)}\mqfi_{(\ap,\ap)}+\beta & 0  & \ldots & 0 & \mqfi_{(\ap,\bq)} && \\
&& 0 & 0 & \ldots & 0 & 0 &&\\
&& \vdots & \vdots& \ddots &\vdots&  \vdots &&\\
&& 0 & 0 & \ldots & 0 & 0 &&\\
&& \mqfi_{(\bq,\ap)} & 0  &  \ldots  & 0  & \frac{1}{(\nparams-1)}\mqfi_{(\bq,\bq)}+\beta &&\\
&&&&&&& 0_{\ddots} & \\
&&&&&&&& 0
\end{array}
\right] - (\nparams-1)\beta \II \nonumber\\
 &\overset{c}{=} \frac{1}{2}(2\mqfi + 2(\nparams-1)\beta\II) - (\nparams-1)\beta \II = \mqfi ,\nonumber
\end{align}
where $(a)$ follows by taking the expectation over the random pair of coordinates $(\ap,\bq)$, $(b)$ follows from Lemma 3, and $(c)$ follows by expanding the summation and noting the following arguments. The non-zero off-diagonal terms arise in the summation only twice, whereas the non-zero diagonal terms appear in the summation $2(\nparams\!-\!1)$ times. 
This completes the proof of Theorem 2.
\section{Derivation of the Gradient of Per-Sample Expected Loss.}
\noindent Recall, from \eqref{eqn:QNNgrad}, the derivative of $U(\param)$ with respect to the ($\a^\textth$ layer, $\p^\textth$ qubit) parameter $\param_{(\a,\p)}$ is given as:
\begin{align*}
        \partial_{\ap} U(\param) \deq \frac{\partial U(\param)}{\partial {\param_{(\a,\p)}}} 
        &= -\frac{i}{2} W_{[L:\a]} \Sigmaap \ W_{(\a:1]}. 
\end{align*}
Using the above derivative, compute the following expression:
\begin{align}
    (\partial_{\ap} U(\param)) \Phi_x U^{\dagger}(\param) + U(\param) \Phi_x  (\partial_{\ap}U^{\dagger}(\param)) &= -\frac{i}{2} W_{[L:\a]}\Sigmaap \ \Phi^{\a}_x W^\dagger_{[\a:L]}+ \frac{i}{2}W_{[L:\a]}\Phi^{\a}_x \Sigmaap  \ W^\dagger_{[\a:L]} \nonumber \\
    &= -\frac{i}{2}W_{[L:\a]} [\Sigmaap, \Phi^\a_x]W^\dagger_{[\a:L]}. \label{eqn:gradcomm}
\end{align}

Thus, using \eqref{eqn:gradcomm}, the gradient of per-sample expected loss can be written as:
\begin{align}
    \frac{\partial \Loss(\param,\state,y)}{\partial \param_{(\a,\p))}} &= \sum_{\yhat \in \calY} \ell(y,\yhat) \Tr\set{\Lambda_{\yhat}((\partial_{\ap} U(\param)) \Phi U^{\dagger}(\param) + U(\param) \Phi  (\partial_{\ap}U^{\dagger}(\param)))} \nonumber\\
    &= \sum_{\yhat \in \calY} -\frac{i}{2} \ell(y,\yhat) \Tr\set{\Lambda_{\yhat}W_{[L:\a]} [\Sigmaap, \Phi^\a_x]W^\dagger_{[\a:L]}}.
\end{align}
This completes the derivation. 

\section{Proof of Lemma 4.}
We first prove the following lemma.

\begin{lemma}\label{lem:param_rule}
Consider a Hermitian operator $\bfA$ on Hilbert space $\calH$, such that $\bfA^2 = \II$. Then, for any operator $\rho$ on $\calH$ the following holds:
\[[\bfA,\rho] = i(e^{-i\pi \bfA/4} \ \rho \ e^{i\pi \bfA/4} - e^{i\pi \bfA/4} \ \rho \ e^{-i\pi \bfA/4}).\]
\end{lemma}
\begin{proof}
    The proof follows from \cite{Mitarai2018}. However, for convenience, we provide a succinct proof below.
    With the definition of $\bfA$, note that $\bfA^2 = \II$. Therefore, using Taylor expansion, we have, 
$$e^{i\beta \bfA} = \cos(\beta)\II 
+ i \sin(\beta) \bfA \ \text{ for all } \beta \in [0,2\pi).$$
Next, we simplify the following expression:
\begin{align*}
    i(e^{-i\pi \bfA/4} \ \rho \ e^{i\pi \bfA/4} - e^{i\pi \bfA/4} \ \rho \ e^{-i\pi \bfA/4}) &= \frac{i}{2} \big((\II-i\bfA)\rho(\II+i\bfA) - (\II+i\bfA)\rho(\II-i\bfA)\big) = [\bfA,\rho].
\end{align*}
This completes the proof of the above Lemma \ref{lem:param_rule}.
\end{proof}

With the intention of employing the above lemma and considering the definition of $\bfV$, we have 
\begin{align*}\bfV (\rho \tensor \ketbra{+})\bfV^{\dagger} &= \frac{1}{2}\big(e^{i\pi \bfA/4} \ \rho \ e^{-i\pi \bfA/4} \tensor |0\>\<0| + 
e^{i\pi \bfA/4} \ \rho \ e^{i\pi \bfA/4}  \tensor |0\>\<1|  \\
& \hspace{1in} +
e^{-i\pi \bfA/4} \ \rho \ e^{-i\pi \bfA/4}  \tensor |1\>\<0|+
e^{-i\pi \bfA/4} \ \rho \ e^{i\pi \bfA/4}  \tensor |1\>\<1|\big).
\end{align*}
Next, using the above expression, we compute the following expectation:
\begin{align}
   2i \Tr\{O \bfV \Tilde{\rho} \bfV^\dagger\} &= i\big(
    \Tr\{\bfB e^{i\pi \bfA/4} \ \rho \ e^{-i\pi \bfA/4} 
    - \bfB e^{-i\pi \bfA/4} \ \rho \ e^{i\pi \bfA/4}\} 
    \big)  \overset{a}{=} \Tr\{\bfB [\rho,\bfA]\} \overset{b}{=} \Tr\{ [\bfA,\bfB]\rho\}, \nonumber
\end{align}
where $(a)$ follows from Lemma \ref{lem:param_rule} and $(b)$ follows from the cyclicity of trace. This completes the proof of Lemma 4.


\section{Proof of Lemma 5.}
Without loss of generality, assume $\a<\b$.
Next, consider the following inequalities:
\begin{align}
    \EE[\bfg{(\ap,\bq)}] &= \sum_{\ap \neq \bq} \frac{1}{\nparams(\nparams-1)} \EE\left[\frac{\nparams}{2}(g_{\ap} \bfe_{\ap} + g_{\bq} \bfe_{\bq})\right]\nonumber\\
    &=\sum_{\ap \neq \bq} \frac{1}{2(\nparams-1)} \big(\EE_{\Qxy}\big[\EE[g_{\ap}|\ketphi{x_1},y_1]\big] \bfe_{\ap} + \EE_{\Qxy}\big[\EE[g_{\bq}|\ketphi{x_2},y_2]\big] \bfe_{\bq}\big) \nonumber\\
    &\overset{a}{=}\sum_{\ap \neq \bq} \frac{1}{2(\nparams-1)} \big(\EE_{\Qxy}[\gradL_{\ap}(\ketphi{x_1},y_1)] \bfe_{\ap} + \EE_{\Qxy}[\gradL_{\bq}(\ketphi{x_2},y_2)] \bfe_{\bq}\big) \nonumber\\
     &=\sum_{\ap \neq \bq} \frac{1}{2(\nparams-1)} \big(\gradL_{\ap} \bfe_{\ap} + \gradL_{\bq} \bfe_{\bq}\big) = \gradL, \nonumber
\end{align}
where $(a)$ follows from the construction of $g_{\ap} \eqand g_{\bq}$.
\section{Proof of Theorem 3.}
We begin by considering Assumption 1 and using the update rule (21). Consider the following inequalities:
    \begin{align*}
        \Loss(\param^{(\ttt+1)}) &\overset{}{\leq} \Loss(\param^{(\ttt)}) - \eta \gradL(\param^{(\ttt)})^{\ttT} \big|\Zbartij\big|^{-1}\bfg^\ttt(i,j) + \eta^2 \frac{\mathsf{L}_2}{2} \big\||\Zbartij|^{-1}\bfg^\ttt(i,j)\big\|^2\\
        &\overset{}{=} \Loss(\param^{(\ttt)}) - \frac{\beta}{(\nparams\!-\!1)\mathsf{L}_2}\gradL^\ttt(i,j)^{\ttT}\Big|\Tilde{\bfZ}^\ttt(i,j) - \frac{2\beta}{\nparams}\II\Big|^{-1}(\git\bfe_i + \gjt \bfe_j) \\
        &\hspace{2.6in}+ \frac{\beta^2}{2(\nparams\!-\!1)^2\mathsf{L}_2} \big\||\Ztildeij^{\ttt} - (2\beta/\nparams)\II_2|^{-1}[\git,\gjt]\big\|^2,
    \end{align*}
where $\gradL^\ttt(i,j)\deq (\gradL_i(\param^{(\ttt)}) \ \bfe_i + \gradL_j(\param^{(\ttt)})\ \bfe_j)$. After taking the expectation of both sides with respect to $\tilde{\bfZ}_{[\text{I},\text{J}]}^{\ttt}, g_\text{I}^{\ttt}, g_\text{J}^{\ttt}$ conditioned on all estimates from previous iterations, we get
\begin{align}
\EE_{(g_\text{I}^{\ttt},\ g_\text{J}^{\ttt},\tilde{\bfZ}_{[\text{I},\text{J}]}^{\ttt})}&[\Loss(\param^{(\ttt+1)})] \nonumber\\
&\overset{}{\leq} \Loss(\param^{(\ttt)})-\frac{\beta}{(\nparams\!-\!1)\mathsf{L}_2} \underbrace{\sum_{i \neq j}  \frac{1}{c(c-1)}\gradL^\ttt(i,j)^{\ttT}\EE\bigg[\Big|\Tilde{\bfZ}^\ttt(i,j) - \frac{2\beta}{\nparams}\II\Big|^{-1}(\git\bfe_i + \gjt \bfe_j)\bigg]}_{\text{T}_1} \nonumber \\
&\hspace{1.1in} + \frac{\beta^2}{2(\nparams\!-\!1)^2\mathsf{L}_2} \underbrace{\sum_{i \neq j}  \frac{1}{c(c-1)} \EE\big[\||\Ztildeij^{\ttt} - (2\beta/\nparams)\II_2|^{-1}[\git,\gjt]\|^2\big]}_{\text{T}_2}. \label{eqn:T1_T2_2QNCDconv}
\end{align}
Now, we simplify the term $\text{T}_2$ as:
\begin{align*}
    \text{T}_2 & \leq \sum_{i \neq j}  \frac{1}{c(c-1)} \EE\big[\||\tilde{\bfZ}^{\ttt}_{[i,j]} - (2\beta/\nparams)\II_2|^{-1}\|^2\big]\EE[\|[\git,\gjt]\|^2] \leq \sum_{i \neq j}  \frac{1}{c(c-1)}{\alpha^2} = \alpha^2,
\end{align*}
where the first inequality follows from the definition of the spectral norm and the fact that independent quantum samples are used to construct $\tilde{\bfZ}_{[i,j]},g_i,\eqand g_j$, and the second inequality follows by defining $\alpha^2 = \max_{(i,j)} \EE\big[\||\Ztildeij^{\ttt}- (2\beta/\nparams)\II_2|^{-1}\|^2\big]\EE[\|[\git,\gjt]\|^2]$ and considering the following arguments. From Remark 2, there exists a $\beta>0$ for every $\nparams>2$ such that the  $2 \times 2$ matrix $(\tilde{\bfZ}_{[i,j]}^{\ttt} - ({2\beta}/{\nparams})\II_2)$ has positive eigenvalues for all possible measurements outcomes. Moreover, given a set of measurement outcomes, the eigenvalues of this $2 \times 2$ sub-matrix eventually saturate to a value independent of $\nparams$ and solely dependent on $\beta > 0$. 
This occurs because, as $\nparams$ increases, the diagonal entries of this $2 \times 2$ matrix are predominantly governed by $\beta$. Therefore, the spectral norm of the inverse of this $2 \times 2$ sub-matrix, i.e., the inverse of its minimum eigenvalue, is bounded. Furthermore, for a bounded loss function $\ell(y, \hat{y})$, the partial derivative estimators $g^\ttt_i \eqand g^\ttt_j$ are also bounded. This implies the product of the expectation of these estimators is also bounded.

On a similar note, we can rewrite $\text{T}_1$ as:
\begin{align*}
    \text{T}_1 &
    \overset{}{=} \sum_{i \neq j}  \frac{1}{c(c-1)}{\gradL^{\ttt}(i,j)}^{\ttT}\EE\Big[\Big|\tilde{\bfZ}^\ttt(i,j) - \frac{2\beta}{\nparams}\II\Big|^{-1}\Big]\gradL^{\ttt}(i,j)\\
    &\overset{a}{\geq} \sum_{i \neq j}  \frac{1}{c(c-1)}{\gradL^{\ttt}(i,j)}^{\ttT}\Big(\EE\Big[\Big|\tilde{\bfZ}^\ttt(i,j) - \frac{2\beta}{\nparams}\II\Big|\Big]\Big)^{\!-1}\gradL^{\ttt}(i,j)\\
    &\overset{b}{=} \sum_{i \neq j}  \frac{1}{c(c-1)}\gradL(\param^{(\ttt)})^{\ttT}_{[i,j]}\left(\mqfitildeij +\frac{2\beta}{\nparams}\II\right)^{-1}\gradL(\param^{(\ttt)})_{[i,j]} \\
    &\overset{c}{\geq}  \Bigg( \sum_{i \neq j}\frac{1}{c(c-1)}\gradL^{\ttt}(i,j)\Bigg)^{\!\!\ttT}\!\!\Bigg(\sum_{i \neq j}\frac{1}{c(c-1)} \left(\mqfitildeij +\frac{2\beta}{\nparams}\II\right)\Bigg)^{\!\!-1}\!\!\Bigg( \sum_{i \neq j}\frac{1}{c(c-1)}\gradL^{\ttt}(i,j)\Bigg)\\
    &=\frac{2}{c}\gradL(\param^{(\ttt)})^{\ttT} \Bigg(\frac{1}{(c-1)}\mqfi(\param) + \frac{2(c-2)}{c}\beta\II\Bigg)^{-1} \gradL(\param^{(\ttt)}) \\
    &\overset{d}{\geq}\frac{2}{c}\gradL(\param^{(\ttt)})^{\ttT} \Bigg(\frac{1}{(c-1)}\mqfi(\param) + 2\beta\II\Bigg)^{-1} \gradL(\param^{(\ttt)}) \\
    &\overset{e}{\geq}\frac{4\mubar}{c}(\Loss(\param^{(\ttt)})-\Loss(\paramstar)) .
\end{align*}
In the above inequalities, $(a)$ follows because 
the inverse is a convex operator function over the space of positive definite operator \cite{nordstrom2011convexity},  $(b)$ follows from Remark 2 and by defining $\mqfitildeij \deq \EE[\tilde{\bfZ}^\ttt(i,j)] - \big(\frac{4\beta}{\nparams}\big)\II(i,j)$, where $\II(i,j)$ is a projection operator with all zero elements except at diagonals corresponding to coordinates $i\eqand j$, $(c)$ follows from Kiefer inequality \cite[Lemma 3.2]{kiefer1959optimum}, $(d)$ follows from the fact that $\frac{(\nparams-2)}{\nparams}\leq 1 $ for all $\nparams>2$ and inverse is operator monotone decreasing function on the space of positive definite matrices, and $(e)$ follows from Assumption 2.
After putting the value of $\text{T}_1 \eqand \text{T}_2$ in \eqref{eqn:T1_T2_2QNCDconv} and taking expectation with respect to all estimates from previous iterations, we get 
\begin{align*}
    \EE[\Loss(\param^{(\ttt+1)}) - \Loss(\paramstar)]
&\overset{}{\leq} \bigg(1-\frac{4\mubar\beta}{\nparams(\nparams\!-\!1)\mathsf{L}_2}\bigg)\EE[(\Loss(\param^{(\ttt)}) - \Loss(\paramstar))] + \frac{\alpha^2\beta^2}{2(\nparams\!-\!1)^2\mathsf{L}_2}
\end{align*}
Finally, applying this inequality recursively, we get 
\begin{align*}
\EE[\Loss(\param^{(\ttt)}) - \Loss(\paramstar)]
&\overset{}{\leq} \bigg(1-\frac{4\mubar \beta}{\nparams(\nparams\!-\!1)\mathsf{L}_2}\bigg)^{\ttt}(\Loss(\param^{(0)}) - \Loss(\paramstar)) + \frac{\alpha^2\beta^2}{2(\nparams\!-\!1)^2\mathsf{L}_2} \sum_{k = 0}^{\ttt-1} \Big(1-\frac{4\mubar\beta}{\nparams(\nparams\!-\!1)\mathsf{L}_2}\Big)^k\\
\EE[\Loss(\param^{(\ttt)}) - \Loss(\paramstar)]
&\overset{}{\leq} \bigg(1-\frac{4\mubar\beta}{\nparams^2\mathsf{L}_2}\bigg)^{\ttt}(\Loss(\param^{(0)}) - \Loss(\paramstar)) +\frac{\alpha^2\beta}{4\mubar}.
\end{align*}
This completes the proof of Theorem 3.


\section{Details of Numerical Implementations}
\noindent \textbf{Experiment Setup.} To evaluate the 2-QNSCD performance, we utilize the Pennylane v0.34.0 open-source library \cite{bergholm2018pennylane} for implementing Algorithm 3. A constant learning rate of $\eta = 2.5 \times 10^{-3}$ is used for all experiments. 
The initial parameters $\param^{(0)}$ are chosen randomly and independently from a uniform distribution over $[0,2\pi)^{\nparams}$. The measurement used in the readout qubits has two outcomes $\set{+1,-1},$ each measured along the computational basis as $\Lambda \deq \set{\Lambda_{+1},\Lambda_{-1} = (\II-\Lambda_{+1})},$ where 
\begin{equation*}
    \Lambda_{+1} = \ketbra{000}+ \ketbra{011}+\ketbra{101}+\ketbra{110}, \quad \text{for } d= 3, \eqand
\end{equation*}  
\begin{equation*}
    \Lambda_{+1} = \sum_{j=0:(j\!\! \texttt{ mod }\!\! 2 = 0)}^{2^d-1}  \ketbra{\texttt{bin}(j)}, \quad \text{for } d= 4,5,6.
\end{equation*}

For $d=3$, this amounts to performing a complete measurement in the binary computational basis and then deciding $\yhat=+1$ if the number of $+1$ observed is even. Otherwise, this amounts to performing a complete measurement in the decimal computational basis and then deciding $\yhat = +1$ if the even outcome is observed. 

\noindent\textbf{PQC Setup.} We consider four cases involving 3, 4, 5, and 6 qubits. 
The parameters are represented using the following convention:  $\bomega_i = \param_{(i/d,i\%d)}$, where $d$ is the number of qubits, $i/d$ represents the integer (floor) division of 
$i$ by $d$, and $i\%d$ ($i$ modulo $d$) gives the remainder of 
$i$ divided by $d$.

\vspace{5pt}
\begin{itemize}
    \item 
    In the 3-qubit configuration, the PQC is composed of three layers of $R_Y$ rotation gates, with each $R_Y$ layer immediately followed by a fixed entangling layer that connects the qubits.
    
\vspace{15pt}
$3$-Qubit, $L=3$
\begin{quantikz}[column sep=0.3cm]
 & \gate{R_{Y}(\bomega_0)} & \ctrl{1} &\! \qw & \gate{R_{Y}(\bomega_3)} & \ctrl{1} &\! \qw & \gate{R_{Y}(\bomega_6)} & \ctrl{1} &\! \qw&\qw  \\
\lstick{} & \gate{R_{Y}(\bomega_1)} & \targ{} & \ctrl{1} & \gate{R_{Y}(\bomega_4)} & \targ{} & \ctrl{1} & \gate{R_{Y}(\bomega_7)} & \targ{} & \ctrl{1} &\qw \\
\lstick{} & \gate{R_{Y}(\bomega_2)} &\! \qw & \targ{} & \gate{R_{Y}(\bomega_5)} &\! \qw & \targ{} & \gate{R_{Y}(\bomega_8)} &\! \qw & \targ{}&\qw  \\
\end{quantikz}.
\item In the 4-qubit case, the PQC consists of two layers where each layer applies a combination of $R_Y$ and $R_Z$ rotations to every qubit, and each of these parameterized layers is followed by an entangling layer.

\vspace{15pt}
$4$-Qubit, $L=4$
\begin{quantikz}
 & \gate{R_{Y}(\bomega_0)} & \gate{R_{Z}(\bomega_4)} & \ctrl{1} &\! \qw & \targ{} & \gate{R_{Y}(\bomega_8)} & \gate{R_{Z}(\bomega_{12})} &\ctrl{2} & \!\qw &\qw \\
& \gate{R_{Y}(\bomega_1)} & \gate{R_{Z}(\bomega_5)} & \targ{} & \ctrl{1} &\! \qw & \gate{R_{Y}(\bomega_9)} & \gate{R_{Z}(\bomega_{13})} &\! \qw & \targ{} &\qw  \\
& \gate{R_{Y}(\bomega_2)} & \gate{R_{Z}(\bomega_6)} & \ctrl{1} & \targ{} &\! \qw & \gate{R_{Y}(\bomega_{10})} & \gate{R_{Z}(\bomega_{14})} & \targ{} &\! \qw  &\qw \\
& \gate{R_{Y}(\bomega_{3})} &\gate{R_{Z}(\bomega_{7})} & \targ{} &\! \qw &\ctrl{-3} &\gate{R_{Y}(\bomega_{11})}  & \gate{R_{Z}(\bomega_{15})} &\! \qw & \ctrl{-2} &\qw   \\
\end{quantikz}

For the 5-qubit and 6-qubit cases, we investigate two distinct circuit architectures for each. 

\item $5$-Qubit, $L=6$ (PQC-1)

\vspace{5pt}

\begin{quantikz}[column sep=0.25cm]
 & \gate{R_{Y}(\bomega_0)} & \gate{R_{Z}(\bomega_5)} & \! \qw & \gate{R_{Y}(\bomega_{10})} & \gate{R_{Z}(\bomega_{15})} &\ctrl{1} & \targ{} & \gate{R_{Y}(\bomega_{20})} & \gate{R_{Z}(\bomega_{25})} &\!\qw &\!\qw\\
 & \gate{R_{Y}(\bomega_{1})} & \gate{R_{Z}(\bomega_{6})} & \ctrl{1} & \gate{R_{Y}(\bomega_{11})} & \gate{R_{Z}(\bomega_{16})} &\targ{} & \!\qw & \gate{R_{Y}(\bomega_{21})} & \gate{R_{Z}(\bomega_{26})} &\ctrl{1} &\!\qw\\
& \gate{R_{Y}(\bomega_{2})} & \gate{R_{Z}(\bomega_{7})} & \targ{} & \gate{R_{Y}(\bomega_{12})} & \gate{R_{Z}(\bomega_{17})} &\ctrl{1} & \!\qw & \gate{R_{Y}(\bomega_{22})} & \gate{R_{Z}(\bomega_{27})} &\targ{} &\!\qw\\
& \gate{R_{Y}(\bomega_{3})} & \gate{R_{Z}(\bomega_{8})} & \ctrl{1} & \gate{R_{Y}(\bomega_{13})} & \gate{R_{Z}(\bomega_{18})} &\targ{} & \!\qw & \gate{R_{Y}(\bomega_{23})} & \gate{R_{Z}(\bomega_{28})} &\ctrl{1} &\!\qw\\
& \gate{R_{Y}(\bomega_{4})} & \gate{R_{Z}(\bomega_{9})} & \targ{} & \gate{R_{Y}(\bomega_{14})} & \gate{R_{Z}(\bomega_{19})} &\!\qw & \ctrl{-4} & \gate{R_{Y}(\bomega_{24})} & \gate{R_{Z}(\bomega_{29})} &\targ{} &\!\qw\\
\end{quantikz}

\item 
    $5$-Qubit, $L=6$ (PQC-2)

\vspace{5pt}
\begin{quantikz}[column sep=0.3cm]
 & \gate{R_{Y}(\bomega_0)} & \gate{R_{Z}(\bomega_5)} & \gate{R_{Y}(\bomega_{10})} & \!\qw &\ctrl{1} & \targ{}& \gate{R_{Z}(\bomega_{15})}  & \gate{R_{Y}(\bomega_{20})} & \gate{R_{Z}(\bomega_{25})} &\!\qw\\
 & \gate{R_{Y}(\bomega_{1})} & \gate{R_{Z}(\bomega_{6})}& \gate{R_{Y}(\bomega_{11})}  & \ctrl{1} &\targ{} & \!\qw & \gate{R_{Z}(\bomega_{16})} & \gate{R_{Y}(\bomega_{21})} & \gate{R_{Z}(\bomega_{26})}  &\!\qw\\
& \gate{R_{Y}(\bomega_{2})} & \gate{R_{Z}(\bomega_{7})} &  \gate{R_{Y}(\bomega_{12})}  & \targ{} &\ctrl{1} & \!\qw & \gate{R_{Z}(\bomega_{17})}  & \gate{R_{Y}(\bomega_{22})} & \gate{R_{Z}(\bomega_{27})}  &\!\qw\\
& \gate{R_{Y}(\bomega_{3})} & \gate{R_{Z}(\bomega_{8})} &  \gate{R_{Y}(\bomega_{13})} & \ctrl{1} &\targ{} & \!\qw  & \gate{R_{Z}(\bomega_{18})} & \gate{R_{Y}(\bomega_{23})} & \gate{R_{Z}(\bomega_{28})}  &\!\qw\\
& \gate{R_{Y}(\bomega_{4})} & \gate{R_{Z}(\bomega_{9})} &  \gate{R_{Y}(\bomega_{14})}  & \targ{} &\!\qw & \ctrl{-4} & \gate{R_{Z}(\bomega_{19})}  & \gate{R_{Y}(\bomega_{24})} & \gate{R_{Z}(\bomega_{29})} &\!\qw\\
\end{quantikz}

\item 
    $6$-Qubit, $L=6$ (PQC-1)

\vspace{5pt}
\begin{quantikz}[column sep=0.25cm]
 & \gate{R_{Y}(\bomega_0)} & \gate{R_{Z}(\bomega_6)} & \ctrl{1} & \gate{R_{Y}(\bomega_{12})} & \gate{R_{Z}(\bomega_{18})} &\!\qw & \targ{} & \gate{R_{Y}(\bomega_{24})} & \gate{R_{Z}(\bomega_{30})} &\ctrl{1} &\!\qw\\
 & \gate{R_{Y}(\bomega_1)} & \gate{R_{Z}(\bomega_7)} & \targ{} & \gate{R_{Y}(\bomega_{13})} & \gate{R_{Z}(\bomega_{19})} &\ctrl{1} & \!\qw & \gate{R_{Y}(\bomega_{25})} & \gate{R_{Z}(\bomega_{31})} &\targ{} &\!\qw\\
 & \gate{R_{Y}(\bomega_{2})} & \gate{R_{Z}(\bomega_{8})} & \ctrl{1} & \gate{R_{Y}(\bomega_{14})} & \gate{R_{Z}(\bomega_{20})} &\targ{} & \!\qw & \gate{R_{Y}(\bomega_{26})} & \gate{R_{Z}(\bomega_{32})} &\ctrl{1} &\!\qw\\
& \gate{R_{Y}(\bomega_{3})} & \gate{R_{Z}(\bomega_{9})} & \targ{} & \gate{R_{Y}(\bomega_{15})} & \gate{R_{Z}(\bomega_{21})} &\ctrl{1} & \!\qw & \gate{R_{Y}(\bomega_{27})} & \gate{R_{Z}(\bomega_{33})} &\targ{} &\!\qw\\
& \gate{R_{Y}(\bomega_{4})} & \gate{R_{Z}(\bomega_{10})} & \ctrl{1} & \gate{R_{Y}(\bomega_{16})} & \gate{R_{Z}(\bomega_{22})} &\targ{} & \!\qw & \gate{R_{Y}(\bomega_{28})} & \gate{R_{Z}(\bomega_{34})} &\ctrl{1} &\!\qw\\
& \gate{R_{Y}(\bomega_{5})} & \gate{R_{Z}(\bomega_{11})} & \targ{} & \gate{R_{Y}(\bomega_{17})} & \gate{R_{Z}(\bomega_{23})} &\!\qw & \ctrl{-5} & \gate{R_{Y}(\bomega_{29})} & \gate{R_{Z}(\bomega_{35})} &\targ{} &\!\qw\\
\end{quantikz}

\item 
    $6$-Qubit, $L=8$ (PQC-2)

\vspace{5pt}
\hspace{-0.25in}\begin{quantikz}[column sep=0.175cm]
 & \gate{R_{Y}(\bomega_0)} & \gate{R_{Z}(\bomega_6)} &  \gate{R_{Y}(\bomega_{12})} & \gate{R_{Z}(\bomega_{18})} &\ctrl{1} &\!\qw & \targ{} & \gate{R_{Y}(\bomega_{24})} & \gate{R_{Z}(\bomega_{30})}& \gate{R_{Y}(\bomega_{36})} & \gate{R_{Z}(\bomega_{42})} &\!\qw\\
 & \gate{R_{Y}(\bomega_1)} & \gate{R_{Z}(\bomega_7)}  & \gate{R_{Y}(\bomega_{13})} & \gate{R_{Z}(\bomega_{19})} & \targ{} &\ctrl{1} & \!\qw & \gate{R_{Y}(\bomega_{25})} & \gate{R_{Z}(\bomega_{31})}& \gate{R_{Y}(\bomega_{37})} & \gate{R_{Z}(\bomega_{43})} &\!\qw\\
 & \gate{R_{Y}(\bomega_{2})} & \gate{R_{Z}(\bomega_{8})}  & \gate{R_{Y}(\bomega_{14})} & \gate{R_{Z}(\bomega_{20})} & \ctrl{1}&\targ{} & \!\qw & \gate{R_{Y}(\bomega_{26})} & \gate{R_{Z}(\bomega_{32})} & \gate{R_{Y}(\bomega_{38})} & \gate{R_{Z}(\bomega_{44})} &\!\qw\\
& \gate{R_{Y}(\bomega_{3})} & \gate{R_{Z}(\bomega_{9})} & \gate{R_{Y}(\bomega_{15})} & \gate{R_{Z}(\bomega_{21})} & \targ{} &\ctrl{1} & \!\qw & \gate{R_{Y}(\bomega_{27})} & \gate{R_{Z}(\bomega_{33})}& \gate{R_{Y}(\bomega_{39})} & \gate{R_{Z}(\bomega_{45})} &\!\qw\\
& \gate{R_{Y}(\bomega_{4})} & \gate{R_{Z}(\bomega_{10})} & \gate{R_{Y}(\bomega_{16})} & \gate{R_{Z}(\bomega_{22})} & \ctrl{1} &\targ{} & \!\qw & \gate{R_{Y}(\bomega_{28})} & \gate{R_{Z}(\bomega_{34})}& \gate{R_{Y}(\bomega_{40})} & \gate{R_{Z}(\bomega_{46})}  &\!\qw\\
& \gate{R_{Y}(\bomega_{5})} & \gate{R_{Z}(\bomega_{11})}  & \gate{R_{Y}(\bomega_{17})} & \gate{R_{Z}(\bomega_{23})} & \targ{} &\!\qw & \ctrl{-5} & \gate{R_{Y}(\bomega_{29})} & \gate{R_{Z}(\bomega_{35})} & \gate{R_{Y}(\bomega_{41})} & \gate{R_{Z}(\bomega_{47})} &\!\qw\\
\end{quantikz}
\end{itemize}

\noindent\textbf{Regularization constant $\beta$.} We choose $\beta$ by empirically assessing the positive definiteness of $2\times 2$ sub-matrix $\bar{\bfZ}_{[i,j]}$ for all possible cases of measurement outcomes. If $\beta > 0.643$ for $c=9$, $\beta > 0.572$ for $c=16$, $\beta > 0.536$ for $c=30$, $\beta > 0.5295$ for $c=36$, and $\beta > 0.5218$ for $c=48$, we observe the sub-matrix $\bar{\bfZ}_{[i,j]}$ is positive definite for all possible measurement outcomes. Therefore, we chose a value for $\beta$, which is close to the threshold.

The regularization constant $\beta$ is a hyper-parameter to trade off numerical instabilities for the faithful E-QFIM estimation. A small $\beta$ leads to a faster convergence, while a larger $\beta$ aligns the update closer to the RQSGD approach. However, using a significantly small $\beta$, closer to the threshold, leads to large noisy oscillations in the 2-QNSCD loss function due to numerical instabilities in the inversion of $\bar{\bfZ}_{[i,j]}$. Figure \ref{fig:beta} shows the 3Q Exp2 example (from Fig.~2) for different values of $\beta$ and illustrates how the different values of $\beta$ impact the performance of 2-QNSCD. As $\beta$ increases, 2-QNSCD approaches the performance of 6-RQSGD and then the 2-RQSGD.

\begin{figure}[!htb]
    \centering
    \includegraphics[height=4in, width=0.95\textwidth]{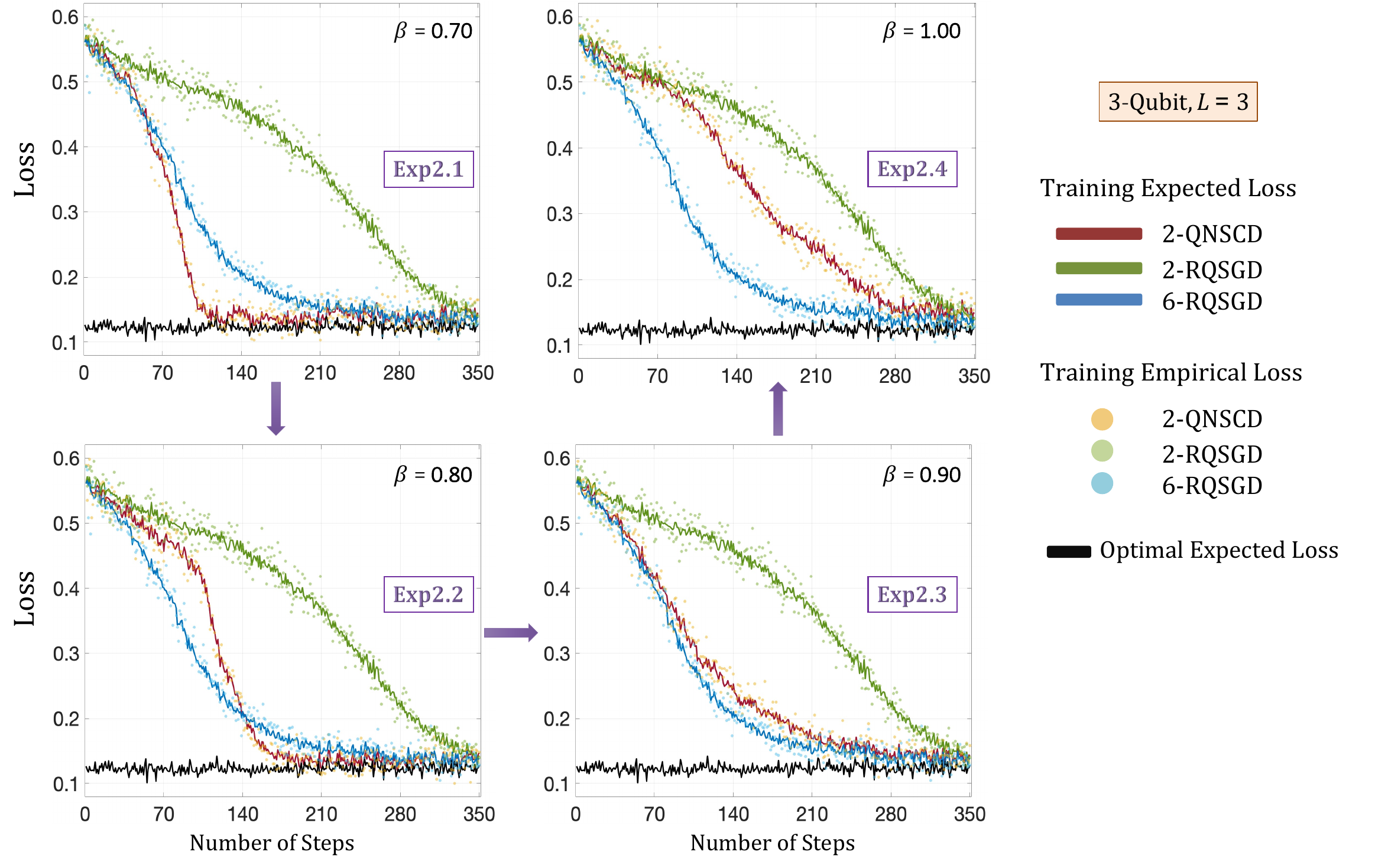}
    \caption{Performance of 2-QNSCD with 3-Qubit for different regularization constant $\beta$.}
    \label{fig:beta}
\end{figure}

\noindent\textbf{RQSGD optimization method.} In 2-RQSGD, three samples are used to estimate the partial derivative of the loss function for each parameter. Let $(\ap,\bq)$ be the pair of coordinates chosen at iteration $\ttt$. 
Then, the estimate for $\gradL_{\ap}(\param^{(\ttt)})$ is calculated as \begin{equation*}
    g_{\ap}^{\ttt} = \frac{1}{3}\sum_{j=1}^3(-1)^{(1+\ttb_j)}\ell(y_j,\yhat_j),\end{equation*} 
    where $y_j'$s are true labels and $(\ttb_j,\yhat_j)'$s are the measurement outcomes of the gradient estimation circuit (see Algorithm 2) corresponding to the three samples used. Similarly, $g_{\bq}^{\ttt}$ is calculated using the remaining three samples. Finally, the unbiased gradient estimator for 2-RQSGD is calculated as given in (20).  In 6-RQSGD, each sample is used to compute the estimate of the partial derivative of the loss function for one parameter, following a procedure similar to Algorithm 2 but applied to six coordinates. Let $(j_1,j_2,\cdots,j_6)$ be the six coordinates chosen at iteration $\ttt$. Then, the unbiased estimator of the gradient for 6-RQSGD is calculated as \begin{equation*}
        \bfg^\ttt = \left(\frac{c}{6}\right)(g_{j_1}^\ttt\bfe_{j_1} + \cdots + g_{j_6}^\ttt\bfe_{j_6}).
    \end{equation*}

\noindent \textbf{Problem of exploding gradient and E-QFIM estimators.} While comparing the performance of 2-QNSCD with RQSGD, note that the gradient estimates grow with $O(c)$ for a bounded loss function $\ell(y,\yhat)$, whereas the elements of the inverse of the E-QFIM estimator diminishes with approximately $O(c^2)$. This implies, for RQSGD, $\bfg$ explodes with $O(c)$, and for 2-QNSCD, the $\Bar{\bfZ}^{-1}\bfg$ roughly diminishes with $O(c)$. As a result, this discrepancy makes it challenging to compare RQSGD and 2-QNSCD directly. This issue of exploding and diminishing gradients has been observed in the training of classical neural networks, particularly in recurrent neural networks \cite{bengio1994learning}. Various methods have been proposed to address this problem, including gradient clipping, normalized parameter initialization, and re-scaling of the gradient \cite{pascanu2013difficulty,glorot2010understanding,goodfellow2016deep}. Similarly, to circumvent the problem of diminishing and exploding estimates, we consider appropriately scaling the estimators with a global constant, ensuring that estimators neither explode nor diminish as $\nparams$ increases. This approach preserves the underlying structure of the estimators while preventing them from becoming unstable. 


\end{document}